\documentclass[letter,fleqn]{cas-sc}

\usepackage[authoryear]{natbib}
\usepackage{amsmath,latexsym,amsbsy,amssymb,amsthm}
\usepackage{algorithm2e}
\usepackage{tikz}
\usetikzlibrary{calc, shapes, arrows.meta, positioning, backgrounds, shapes.geometric}
\usepackage{svg}
\usepackage{subcaption}
\usepackage{setspace}
\usepackage[version=4]{mhchem}

\RestyleAlgo{ruled}

\newtheorem{theorem}{Theorem}
\newtheorem{lemma}[theorem]{Lemma}       
\newtheorem{corollary}[theorem]{Corollary}
 
\newtheorem{fact}{Fact}

\theoremstyle{definition}
\newtheorem{definition}{Definition}

\theoremstyle{remark}
\newtheorem{remark}{Remark}

\theoremstyle{plain} 
\newtheorem{assumption}{Assumption}

\newcommand{\R}{\mathbb{R}}

\begin{document}
\let\WriteBookmarks\relax
\def\floatpagepagefraction{1}
\def\textpagefraction{.001}
\shorttitle{Nonautonomous KKL Observer}
\shortauthors{M. Woelk et~al.}

\title [mode = title]{Neural Luenberger state observer for nonautonomous nonlinear systems}                      
\tnotemark[1]

\tnotetext[1]{This work is supported by American Chemical Society (ACS) under the Petroleum Research Fund (Award \#66911-DNI9) and by National Science Foundation (NSF) under Division of Chemical, Bioengineering, Environmental and Transport Systems (Award \#2414369).}

\author[]{Moritz Woelk}[orcid=0000-0002-5874-2133]
\credit{Methodology; investigation; validation; software; data curation; writing - original draft; writing - review and editing}

\affiliation[]{organization={Department of Chemical and Biomolecular Engineering, North Carolina State University}, 
                city={Raleigh},
                state={NC 27695-7905},
                country={USA}}

\author[]{Jarod Morris}
\credit{Methodology; investigation; validation; software; data curation; writing - original draft; writing - review and editing}

\author[]{Wentao Tang}[orcid=0000-0003-0816-2322]
\cormark[1]
\ead{wentao_tang@ncsu.edu}
\credit{Supervision; project administration; resources; methodology; funding acquisition; writing - review and editing}
\cortext[cor1]{Corresponding author}

\begin{abstract}
This work proposes a method for \textit{model-free} synthesis of a state observer for nonlinear systems with manipulated inputs, \textcolor{black}{where the observer is trained offline using a historical or simulation dataset of state measurements.} We use the structure of the Kazantzis-Kravaris/Luenberger (KKL) observer, extended to nonautonomous systems by adding an additional input-affine term to the linear time-invariant (LTI) observer-state dynamics, which determines a nonlinear injective mapping of the true states. Both this input-affine term and the nonlinear mapping from the observer states to the system states are learned from data using fully connected feedforward multi-layer perceptron neural networks. Furthermore, we theoretically prove that trained neural networks, when given new input-output data, can be used to observe the states with a guaranteed error bound. To validate the proposed observer synthesis method, case studies are performed on a bioreactor and a Williams-Otto reactor.
\end{abstract}


\begin{highlights}
\item Model-free observer synthesis for nonlinear systems with manipulated inputs.
\item Luenberger-like observer containing input-affine terms accounting for input effects.
\item Observer output mapping and input-affine terms learned as neural networks. 
\item Numerical case studies on a bioreactor and a Williams-Otto reactor.
\end{highlights}

\begin{keywords}
State Observation\sep Nonlinear Dynamics \sep Neural Networks
\end{keywords}

\maketitle

\section{Introduction}
\begin{doublespace}  
Contemporary control frameworks for modeling, analysis, and control of nonlinear chemical processes rely fundamentally on state-space representations, whether for model-based methods (e.g., system identification \citep{aastrom1971system}, Koopman operator methods \citep{williams2015data,korda2018linear}), or model-free techniques like reinforcement learning \citep{bertsekas2019reinforcement}. Since the full measurability of the state is rarely satisfied in practice due to physical constraints (e.g., unmeasurable concentrations) and fiscal limitations (e.g., sensor costs), the use of \textit{state observers} that can estimate the states of a system from available plant measurements is necessary. In practice, such observers are useful in control systems by enabling output feedback control \citep{teel1995tools,khalil2014high}, maintaining virtual-physical state alignment in digital twins \citep{pratap2019real,riva2024twin}, and facilitating fault diagnosis through early detection of system anomalies \citep{patton1997observer,chen2007observer}. Formally, a state observer can be defined as an auxiliary system that processes the plant's input-output signals to generate state estimate signals.

State observation for linear systems is classically addressed through the Luenberger observer, which uses a gain matrix on the output error term to stabilize the state estimation error dynamics and guarantees that the estimated state converges asymptotically to the true state \citep{luenberger1964observing}, and the Kalman filter, which optimizes state estimates under stochastic noise by minimizing error covariance \citep{kalman1960}. For nonlinear systems, the most widely implemented approach is the extended Kalman filter (EKF), which applies Kalman filtering to a linearization of the system dynamics about the current state estimate \citep{reif_ekf-based_1998}. Alternatively, high-gain observers use a coordinate transformation to render the system in a canonical triangular form, where convergence is yielded by feeding back the output error with sufficiently large gains so that the estimation error dynamics become exponentially stable \citep{hammouri2002high,khalil2014high}. While effective for specific classes of nonlinear systems, both these approaches suffer from limitations. The EKF's reliance on local linearization restricts its convergence to a small neighborhood of the true state, while high-gain observers exhibit transient peaking phenomena \citep{esfandiari1992output} and large sensitivity to noise. Moreover, both approaches remain fundamentally model-dependent, requiring either an exact state-space representation or precise knowledge of the system dynamics. This requirement limits their applicability to systems where the underlying model is either unknown or complex to derive analytically. Such restrictions motivate the development of data-driven observers that reconstruct states directly from input-output measurements, bypassing the need for explicit system modeling altogether. \textcolor{black}{Related data-driven formulations pursue similar goals through virtual sensing, where measurements are used to reconstruct state or error signals at unmeasured or remote locations of interest \citep{petersen2008kalman}, or through direct filter synthesis, which constructs observer dynamics directly from data without intermediate model identification \citep{novara2012direct,mazzoleni2024comparison}.} 

An observer framework that lends itself well to data-driven techniques is a \emph{Kazantzis-Kravaris/Luenberger (KKL)} observer. Such an observer, developed in the seminal work of \cite{kazantzis1998nonlinear}, extends the classical Luenberger observer to (autonomous) nonlinear systems by using an \emph{injective} nonlinear transformation on states. This transformation maps the original system states to \emph{observer states} that are governed by a linear time-invariant (LTI) system driven by the plant's outputs. Consequently, the observer synthesis problem simplifies to identifying both the nonlinear transformation and its left inverse, which requires solving
(model-based) partial differential equations (PDEs). The KKL framework is particularly powerful due to its applicability to broad classes of nonlinear systems and its well-developed theoretical guarantees and convergence analysis \citep{brivadis2019luenberger,bernard2022observer,pachy2024existence}. The existing KKL observer synthesis results primarily focus on autonomous systems. The autonomous case was extended to actuated controlled systems by \cite{bernard2018luenberger} through two key approaches: (i) retaining the stationary transformation (as in the autonomous case) while augmenting the observer dynamics with additional state-dependent input-affine terms, or (ii) using a time-varying transformation that explicitly accounts for exogenous inputs. While theoretically appealing, its practical computation remains challenging if the knowledge of an injective mapping that transforms the original nonlinear system into the specific structure required by the KKL framework is unavailable. However, deriving this mapping poses difficulties due to the need to solve a system of model-based PDEs \citep{bernard2018luenberger}, and even when such a mapping is known, computing its left inverse to recover state estimates in the original coordinates remains nontrivial, due to its dependence on model knowledge.

This paper presents a \textit{model-free} approach utilizing the Kazantzis-Kravaris/Luenberger (KKL) observer framework for \textit{nonautonomous systems}, implemented in a fully data-driven manner \textcolor{black}{through offline training on a stored dataset of state measurements obtained from history or simulations}, without requiring prior knowledge of the system dynamics.  Building on the input-affine extension to the LTI observer dynamics proposed in \cite{bernard2018luenberger} for handling exogenous inputs, we propose that both the input-affine term in the observer dynamics and the pseudoinverse of the injective nonlinear mapping can be learned directly from data using neural networks. This eliminates the dual challenges of first-principles modeling for complex dynamics and analytical derivation of the observer's nonlinear transformation and input-affine term. We refer to our proposed approach as a neural Luenberger observer for nonlinear systems with external inputs (NLOX). 

The powerful approximation capabilities of neural networks make them particularly suitable for learning the nonlinear state transformation by mapping the original state to a high-dimensional observer state space. Recent advances in data-driven state observation began with supervised feedforward neural networks \citep{ramos2020numerical}, later enhanced through unsupervised learning approaches \citep{peralez2021deep} that enable more effective state space exploration during training. Subsequent developments have used physics-informed neural networks (PINNs) \citep{niazi2023learning} that incorporate the governing PDE as a constraint, along with a neural ODE framework \citep{miao2023learning} that formulates the learning problem as a dynamic optimization, minimizing the time-integrated squared observation error while embedding the observer dynamics within the network architecture, or through optimization of latent dynamics for KKL observer tuning \citep{buisson2023towards}. A more recent model-free method by \cite{peralez2024deep} uses a neural network to handle noisy systems by formulating the learning objective as the minimization of the estimation error through the incorporation of arbitrary complex eigenvalues in a switching observer framework. Likewise, in the model-free setting, \cite{tang2024synthesis} uses Lipschitz-bounded neural networks to constrain the generalization error of the state observer according to the Lipschitz constant for robust observation. Most prior approaches focus on systems of the autonomous (unforced) form, with only a few exceptions. For instance, \cite{ramos2020numerical} uses a neural network to learn a nonlinear transformation and its left inverse, though the input-affine term in the observer dynamics remains unlearned, restricting the set of admissible inputs. Similarly, \cite{peralez2021deep} uses neural networks to learn both forward and inverse mappings in an unsupervised framework. \textcolor{black}{However, in the nonautonomous setting considered therein, the system dynamics are assumed to be known, and the observer design remains model-dependent rather than fully data-driven.} Notably, these existing works have primarily examined systems with state dimensions of two. \textcolor{black}{Another class of related methods are neural adaptive observers, which assume a parametric model of the system dynamics and use neural networks to approximate unknown nonlinearities within this model structure \citep{choi2001adaptive,chen2017neural}. These approaches rely on parameter adaptation laws, often requiring persistent excitation, to estimate the unknown nonlinearities and ensure convergence of the parameter estimates \citep{zeng2025observer}. Thus, adaptive observers rely on parameterizing the system dynamics and estimating those parameters online, and therefore belong to the class of model-based observer designs.} \textcolor{black}{Despite these advances, a fully data-driven model-free observer synthesis framework that explicitly learns both the input-affine term and the nonlinear state reconstruction within a KKL-type structure for nonautonomous systems remains unaddressed.}

Indeed, neural networks can suffer from limitations such as nonconvex training problems. Seeking a convex learning formulation, \cite{tang2023data} demonstrates that topologically equivalent state reconstruction remains possible without direct state measurements, using nonlinear dimensionality reduction of observer states to recover a diffeomorphic representation that preserves system trajectories and dynamic relationships. In recent works, the learning of a Koopman operator and its spectral properties for the data-driven synthesis of nonlinear state observers has been studied \citep{ye2025edmd,Ni_tang2025,tang2025koopman}. While the benefit of global linearization brought by the Koopman operator is of theoretical value, the definition of the Koopman operator strongly relies on the prior knowledge of the qualitative properties of the drift system  (i.e., the autonomous past or the unforced system). Comparatively, we still consider the neural network-based learning approach as a practically easy-to-implement one, leaving its afore-mentioned theoretical limitations to possible ad-hoc hyperparameter tuning and validation efforts.

To this end, we use multi-layer feedforward neural networks for their architectural simplicity, enabling \textit{model-free} simultaneous learning of both the nonlinear state transformation and the input-affine correction term for nonautonomous systems with inputs. The rest of this paper is organized as follows. Section \ref{Sec:Prelim_Problem} develops the problem setting and theoretical preliminaries for KKL observer design, including the input-affine extension for nonautonomous systems. The proposed neural network architecture for learning both the state mapping and input correction term is presented in Section \ref{Sec: NeuralApproach}, with implementation details and theoretical probabilistic bounds on the performance of NLOX derived in Section \ref{Sec: Gen_error}. Section \ref{Sec:CaseStudies} validates the approach through two chemical process case studies: (i) a constant-volume bioreactor and (ii) a Williams-Otto reactor, demonstrating accurate state reconstruction from data alone. Concluding remarks and future directions are discussed in Section \ref{Sec:Conclusion}.

\section{Preliminaries and problem settings}\label{Sec:Prelim_Problem}
Here, we consider a continuous-time nonlinear dynamical system in the form of 
\begin{equation}\label{eq: Input-Affine Dynamic System}
        \dot{x}(t) = f(x(t))+g(x(t))u(t), \qquad y(t)=h(x(t))
\end{equation}
where $x(t)\in \mathcal{X}\subseteq \R^{n_x}$ is the vector of system states, $u(t)\in\mathcal{U}\subseteq\R^{n_u}$ is the control inputs, and $y(t)\in\R^{n_y}$ represents the outputs defined on $t\in[0,+\infty)$. The model functions $f:\mathcal{X} \to \R^{n_x}$, $g:\mathcal{X} \to \R^{n_x\times n_u}$, and $h:\mathcal{X} \to \R^{n_y}$ are unknown but assumed to be smooth on $\mathcal{X}$ and guarantee the existence and uniqueness of the solution for all $t\geq0$. Let $\mathcal{X}_0\subseteq \R^{n_x}$ be a subset of $\mathcal{X}$ representing the collection of all possible initial conditions for \eqref{eq: Input-Affine Dynamic System}. 
\begin{assumption}\label{Assum:Solvability}
    For any $x(0)\in\mathcal{X}_0$ and $u(\cdot)\in \mathrm{L}^\infty\left([0,+\infty),\mathcal{U} \right)$, the solution to \eqref{eq: Input-Affine Dynamic System} exists for all $t\in[0,+\infty)$ and satisfies $x(t)\in\mathcal{X}$, where $\mathcal{X}$ is compact.
\end{assumption}

For notational simplicity, we will hereafter omit explicit time dependence (e.g., writing $x$ instead of $x(t)$) unless desirable for clarity. Along with Assumption \ref{Assum:Solvability} to guarantee the existence of solutions, we must also impose a stability requirement. Since the system is subject to external inputs, we require that its states remain bounded for any bounded input, a property known as input-to-state stability (ISS) \citep{sontag1995characterizations}. However, our objective of deriving a data-driven state observer necessitates a stronger, incremental form of stability. Since the observer must reconstruct the state trajectory from input-output data, to ensure that this reconstruction is well-behaved, a small perturbation in the input signal must not cause a large deviation in the resulting state trajectory. Therefore, we strengthen the ISS requirement to incremental input-to-state stability ($\delta$ISS) \citep{angeli2002lyapunov}. Based on the preceding arguments, we now formalize the following standing assumption on the system's stability, which Section \ref{Sec: Gen_error} will show is essential for proving the generalization bound.
 
\begin{assumption}\label{Assum:deltaISS}
    The system \eqref{eq: Input-Affine Dynamic System} is \emph{incrementally input-to-state stable} ($\delta$ISS) with respect to the input $u$. That is, there exists a class-$\mathscr{KL}$ function $\beta$ and a class-$\mathscr{K}$ function $\gamma$ such that for any two initial conditions $x(0),~ x'(0) \in \mathcal{X}$ and any two bounded input signals $u(\cdot),~ u'(\cdot)\in \mathrm{L}^\infty\left([0,+\infty),\mathcal{U} \right)$, the corresponding solutions $x(t)$ and $x'(t)$ satisfy:
    \begin{equation}\label{eq:deltaISS}
    \left\|x\left(t,x(0),u\right) - x'\left(t,x'(0),u'\right)\right\| \leq \beta\left(\left\| x(0)-x'(0) \right\|, t\right) + \gamma\left(\operatorname*{ess\,sup}_{0 \leq \tau \leq t} \left\|u(\tau) - u'(\tau)\right\|\right), \quad \forall t \geq 0.
    \end{equation}
    Furthermore, the functions $\beta$ and $\gamma$ are bounded by linear functions; i.e., $\exists L_{\beta, t}, L_\gamma \geq 0$ such that $\forall r, t \geq 0$
    $$\beta(r, t)\leq L_{\beta, t}r, \quad \gamma(r)\leq L_\gamma r.$$
\end{assumption}\textcolor{black}{\remark{Assumption \ref{Assum:deltaISS} is introduced to ensure bounded incremental behavior of the plant in the presence of manipulated inputs, as the proposed framework does not involve controller synthesis or input optimization to enforce stability. While $\delta$ISS is difficult to verify exactly from data in a model-free setting, it serves as a standard condition on the system's suitability for learning-based techniques. In practice, $\delta$ISS-like behavior may be empirically assessed by evaluating the sensitivity of state trajectories to perturbations in initial conditions and bounded input variations.}}

In this paper, we design an observer for system \eqref{eq: Input-Affine Dynamic System}, via the Luenberger-like methodology. We first introduce the autonomous case (i.e., $u(t)\equiv 0$) to construct the theoretical framework, subsequently extending these results to systems with manipulated inputs.

\subsection{Autonomous systems}
For autonomous systems, \textcolor{black}{the classical Luenberger observer}, restricted to linear systems, has been extended to nonlinear systems with the KKL observer \citep{kazantzis1998nonlinear}. We adopt the standard assumption that inputs and outputs are known in a causal manner, i.e., at any time $t$, the observer can use only their past or current values $\left(u|_{[0,t]},~ y|_{[0,t]}\right)$. The KKL observer for an autonomous system is defined as:
\begin{subequations}\label{eq:KKL-observer}
\begin{align}
\dot{z}(t) &= Az\left(t\right) + By\left(t\right),\label{eq:2a}\\ 
\hat{x}(t) &= T^\dagger (z(t)). \label{eq:2b}
\end{align}
\end{subequations}
The observer states $z(t)\in\R^{n_z}$ evolve under linear time-invariant (LTI) dynamics, defined by system matrices $A\in\R^{n_z\times n_z}$ and $B\in\R^{n_z\times n_y}$. The pair $(A,B)$ must be controllable, i.e., the controllability matrix $\mathcal{C}:=\left[B,AB,\dots,A^{n_z-1}B  \right]$ must satisfy $\text{rank}(\mathcal{C})=n_z$, and $A$ must be Hurwitz (with all eigenvalues $\left\{\lambda_i\right\}_{i=1}^{n_z}$ residing in the open left half complex plane, Re$(\lambda_i)<0$). To synthesize a KKL observer, the mapping $T:\mathcal{X}  \to \mathbb{R}^{n_z}$ must be injective and continuously differentiable, ensuring the existence of a left inverse $T^\dagger:\mathbb{R}^{n_z}  \to\mathcal{X} $ which maps the observer states $(z)$ to the system state estimates $\left(\hat{x}\right)$. This static mapping must satisfy the following PDE system:
\begin{equation}\label{eq:KKL-PDE}
    \frac{\partial T}{\partial x}f(x)=AT(x)+Bh(x), \quad \forall x\in \mathcal{X},
\end{equation}
where $T^\dagger$ is a left pseudo-inverse of the mapping $T$, such that $T^\dagger(T(x))=x, ~\forall x\in\mathcal{X}$, or $T^\dagger\circ T=I_{n_x}.$ and the notation $\partial T/\partial x$ refers to the Jacobian of $T$. If the mapping $T$ satisfies \eqref{eq:KKL-PDE}, it easily follows that $dT(x(t))/dt=AT(x(t))+By(t)$, and since $A$ is Hurwitz, the dynamics of $z-T(x)$ are asymptotically stable. 
 
The existence of an injective mapping $T$ that satisfies \eqref{eq:KKL-PDE} is guaranteed under the condition of backward distinguishability on the systems' flow. Formally, the property is that for any pair of states, $\left(x_1,x_2\right)\in \mathcal{X}^2$ with $x_1\neq x_2$, there exists some past time horizon $\tau>0$ such that the output trajectories over the past time interval $[-\tau,0]$: $ y_i(\cdot)=h\left(x\left(x_i,\cdot\right)\right)~(i=1,2)$ are distinguishable, where $x$ denotes the system flow map. The sufficiency of this relatively mild observability assumption for the existence of $T$ has been proven in the literature \citep{andrieu2006existence, brivadis2023further}.

For nonlinear systems, the required observer dimension $n_z$ generally needs to exceed the system dimension $n_x$ $\left(n_z>n_x\right)$. It is established that the minimal observer dimension for a single-output system $\left(n_y=1\right)$ becomes $n_z=2n_x+1$. This generalizes to $n_z=n_y\left(2n_x+1\right)$ for systems with $n_y$ outputs. Under a slightly more restrictive assumption, one may choose $n_z=n_y\left(n_x+1\right)$ \citep{andrieu2006existence}. In this case, the observer dynamics can be implemented with $A=\text{diag}\left(-\lambda_1,\dots,-\lambda_{n_z} \right)$ and $B=1_{n_z}$, where $\left\{-\lambda_i\right\}_{i=1}^{n_z}$ are the chosen poles (which can be almost arbitrarily chosen on the left half plane) and $1_{n_z}$ denotes the $n_z$-dimensional vector where each component is equal to 1. Thus, the mapping $T:\R^{n_x}\to\R^{n_z}$ performs a lifting of the states into a higher-dimensional space, while its left-inverse $T^\dagger$ projects a dimensionality reduction back into the original state space. 

Although the construction of $T^\dagger$ is crucial for state observation, its explicit determination relies on solving the model-based PDE system \eqref{eq:KKL-PDE}, rendering it impractical for numerical approximation. Nevertheless, recent advances have demonstrated successful approximation through data-driven methods (e.g., physics-informed neural networks \citep{niazi2023learning}, neural ODEs \citep{miao2023learning}, Lipschitz-bound neural networks \citep{tang2024synthesis}). To this end, we adopt a neural network approximation (learning) of the $T^\dagger$ map, with formulation and implementation details deferred to Section \ref{Sec: NeuralApproach}.

The preceding results hold for autonomous systems, but require modification when exogenous inputs are present $(u\neq0)$. In this case, the observer state $z$, defined by \eqref{eq:KKL-observer}, fails to converge asymptotically under the static mapping $T$ from \eqref{eq:KKL-PDE}. Our work adopts the approach of \cite{bernard2018luenberger}, where the LTI observer dynamics are augmented with additional input-dependent terms, preserving the stationary injective mapping while modifying the observer dynamics. The following subsection formalizes this framework for input-affine systems \eqref{eq: Input-Affine Dynamic System}.

\subsection{Nonautonomous systems}
In this subsection, we consider nonautonomous systems of the input-affine form \eqref{eq: Input-Affine Dynamic System}. For such systems, \cite{bernard2018luenberger} established the existence of a KKL observer with static transformations under two observability notions. The first, termed  \textit{differential observability of the drift system}, is characterized by the injectivity of the mapping $H_d(x)$, which collects successive Lie derivatives of the output $h(x)$ along the drift vector field $f$, as formalized in Definition \ref{def:diff_obs}.

\begin{definition}[\textit{Differential observability of the drift system}]\label{def:diff_obs}
    The drift dynamics of system \eqref{eq: Input-Affine Dynamic System}, namely $\dot{x}=f(x)$, are \textit{weakly differentially observable} of order $d$ on an open subset $\mathcal{D} \subseteq{\R}^{n_x}$ if the function
\begin{equation*}
    H_d(x) = \begin{bmatrix}
        h\left(x\right)& L_fh(x)&\cdots&L_f^{d-1}h(x)
    \end{bmatrix}^\top,
\end{equation*}
is injective on $\mathcal{D}$. Here $L_f^0h(x)=h(x)$, and $L_f^{k+1}h(x)=\frac{\partial}{\partial x}\left(L_f^kh(x)\right)f(x),~\text{for}~ k=0,\dots,d-2.$ If $H_d$ is additionally an immersion $\left(\text{i.e., its Jacobian matrix }\frac{\partial H_d}{\partial x}(x)~ \text{has full rank for all}~ x \in \mathcal{D}\right)$, the system is \textit{strongly differentially observable} of order $d$.
\end{definition}

\noindent The above definition ensures that the state can be reconstructed from output derivatives when no input is applied, providing a baseline for observer design. The second required observability notion, \textit{instantaneous uniform observability}, concerns the ability to distinguish between trajectories of system \eqref{eq: Input-Affine Dynamic System} starting from distinct initial states $x_1(0),x_2(0)\in\mathcal{X}_0$ under any identical input signal $u$. In the following definition, we denote $x\left(x_0,t;u\right)$ as the solutions to the ODEs of system \eqref{eq: Input-Affine Dynamic System} at time $t$ with initial condition $x_0\in\mathcal{X}$ and a control signal $u|_{\left[0,t\right)}$ .

\begin{definition}[\textit{Instantaneous uniform observability}]\label{def:Insta_obs}
    System \eqref{eq: Input-Affine Dynamic System} is \emph{instantaneously uniformly observable} on an open subset $\mathcal{D}\subseteq\mathbb{R}^{n_x}$ if for any pair $\left(x_1,x_2\right)\in\mathcal{D}^2$ with $x_1\neq x_2$, any strictly positive time $\tau^*$, and any input function defined on $\left[0,\tau^*\right)$, 
    \begin{enumerate}
        \item there exists a time $t<\tau^*$ such that $h\left(x\left(x_1,t;u\right)\right)\neq h\left(x\left(x_2,t;u\right)\right)$,
        \item $\left(x\left(x_1,s;u\right),x\left(x_2,s;u\right)\right)\in\mathcal{D}^2$ for all $0 \leq s\leq t$.
    \end{enumerate}
\end{definition}

\noindent The Definitions \ref{def:diff_obs}--\ref{def:Insta_obs} are sufficient to use a static transformation for input-affine systems within the Luenberger observer framework.   

\begin{fact}[\cite{bernard2018luenberger}]\label{fact:Bernard_NonAuto}
    Assume that there is an open set $\mathcal{D}\subseteq \R^{n_x}$ containing $\mathcal{X}$ where \eqref{eq: Input-Affine Dynamic System} is instantaneously uniformly observable, and its drift system is strongly differentially observable of order $n_x$. Suppose that \eqref{eq: Input-Affine Dynamic System} is such that for any $ x_0\in\mathcal{X}_0$, and any $u\in \mathcal{U}$ with $|u(t)|\leq U_*\left(\forall t\in [0,+\infty )\right)$, we have $x\left(x_0,t;u\right)\in\mathcal{X}$. Then there exists a positive \textcolor{black}{$\kappa^*$} where for any \textcolor{black}{$\kappa>\kappa^*$} the following exist:
    \begin{enumerate}
        \item a function $T:\mathbb{R}^{n_x}\rightarrow\mathbb{R}^{n_z}$ $\left(n_z\geq n_x \right)$ that is injective on $\mathcal{X}$ and satisfies the PDE associated with the drift dynamics
        \begin{equation*}
        \frac{\partial T}{\partial x}(x)f(x)=\textcolor{black}{\kappa}~AT(x)+Bh(x), \qquad \forall x\in\mathcal{X},
    \end{equation*}
    where for some strictly positive $\left\{\lambda_i\right\}_{i=1}^{n_z}$, $A$ is the Hurwitz matrix $\text{diag}\left(-\lambda_1,\dots,-\lambda_{n_z} \right)\in\R^{n_z\times n_z}$ and $B=1_{n_z}\in\R^{n_z}$.
    \item a Lipschitz function $\omega: \mathbb{R}^{n_z}\to\R^{n_x\times n_u}$ that satisfies
    \begin{equation*}
        \omega\left(z\right)= \frac{\partial T}{\partial x}\left(x\right)g\left(x\right), \qquad\forall x \in \mathcal{X}.
    \end{equation*}
    Furthermore, for any function $T^\dagger:\mathbb{R}^{n_z}\rightarrow\mathbb{R}^{n_x}$ that is an inverse of the injective function $T$ on $\mathcal{X}$, namely any $T^\dagger$ that satisfies $T^\dagger (T(x))=x,~ \forall x\in\mathcal{X}$, the observer state dynamics
    \begin{equation} \label{eq:KKL-observer-NonAuto}
        \dot{z}=Az+By+\omega(z)u,
    \end{equation}
    guarantees vanishing errors, i.e., $\lim_{t\to+\infty}\left(z-T(x)\right)=0$.
\end{enumerate}
\end{fact}

The scalar gain \textcolor{black}{$\kappa$} in the theorem acts as a convergence rate parameter, where larger values yield faster observer dynamics. While the theorem constructs an injective mapping $T$ through the solution of a PDE, as previously stated, this transformation is typically challenging to obtain analytically. Instead, we focus on determining its left inverse $T^\dagger$ (which directly recovers the state estimate) and the input-affine term $\omega(z)$ that compensates for exogenous inputs. Since $T^\dagger$ and $\omega$ depend on the unknown system dynamics, we develop a neural approximation approach in the following section that learns these components simultaneously from accessible data.

\section{State observation algorithm for NLOX}\label{Sec: NeuralApproach}
In this section, we detail the neural parameterization of the input-affine term $\omega$ and the inverse mapping $T^\dagger$. We adopt a neural network approach due to its universal approximation capabilities, which enable accurate representation of these nonlinear mappings. This capability, combined with sufficient training data and appropriate network initialization, allows the neural network to learn these functions with reasonable accuracy \citep{hornik1989multilayer}. We begin by formalizing how training data are sampled from the system and structured for learning the observer dynamics.
\subsection{Data sampling and collection}\label{sec:data_sampling}
Since data can only be collected on sampling intervals, the continuous-time observer dynamics in \eqref{eq:KKL-observer-NonAuto} are discretized with sampling interval $t_s$. Denoting $k = t/t_s$ as the discrete-time index, the observer dynamics can be approximated through a first-order (forward Euler) discretization as 
\begin{equation}\label{eq:KKL_discrete}
z_{k+1}=z_k +t_s \left(Az_k +By_k+\omega(z_k) u_k\right) ,\quad k=0,\dots,N-1
\end{equation}
where $A, B$ denote the continuous-time observer matrices. As such, the simulated $z$ signal at discretized instants satisfies $z_k\neq z^\star\left(kt_s \right)$, where $z^\star$ is the actual value of $z$ if the continuous-time dynamics are followed. We hence denote $\varepsilon_k=z_k-z^\star\left(kt_s \right)$ as the local discretization error introduced by the forward Euler approximation and the piecewise-constant (pulse) input. This numerical error will be accounted for in the generalization error analysis in Section~\ref{Sec: Gen_error}.
\par First, we collect $M$ sample trajectories of input-output data from the system, each containing $N$ sampling times. The datasets can be written in per-trajectory components as:
\begin{equation}\label{eq:Sample_UXY}
    U_i=\begin{bmatrix}
    \left(u^{(i)}_0\right)^\top\\\left(u^{(i)}_1\right)^\top\\ \vdots \\ \left(u^{(i)}_{N-1}\right)^\top
    \end{bmatrix}\in\R^{N\times n_u},~X_i=\begin{bmatrix}
    \left(x^{(i)}_0\right)^\top\\\left(x^{(i)}_1\right)^\top\\ \vdots \\ \left(x^{(i)}_{N-1}\right)^\top
    \end{bmatrix}\in\R^{N\times n_x},~Y_i=\begin{bmatrix}
        \left(y^{(i)}_0\right)^\top\\\left(y^{(i)}_1\right)^\top\\ \vdots \\ \left(y^{(i)}_{N-1}\right)^\top
    \end{bmatrix}\in\R^{N\times n_y},~i=1,\dots,M,~k=0,\dots,N-1.
\end{equation}
generated offline (through simulations or historical plant operations). The per-trajectory matrices in \eqref{eq:Sample_UXY} can then be concatenated into stacked forms:
\begin{equation*}\label{eq:Full_UYX_Dataset}
    U = \begin{bmatrix}
        \vdots\\\left(u^{(i)}_k\right)^\top\\\vdots\\
    \end{bmatrix}\in\R^{(MN)\times n_u},~X = \begin{bmatrix}
        \vdots\\\left(x^{(i)}_k\right)^\top\\\vdots\\
    \end{bmatrix}\in\R^{(MN)\times n_x},~Y = \begin{bmatrix}
        \vdots\\\left(y^{(i)}_k\right)^\top\\\vdots\\
    \end{bmatrix}\in\R^{(MN)\times n_y}.
\end{equation*}

By Fact \ref{fact:Bernard_NonAuto}, the input sequence $u_k$ is required to be bounded. Accordingly, the excitation signal is generated as a piecewise-constant input, held over each sampling interval $\left[ kt_s,(k+1)t_s\right)$ for each $k$. In this way, we can obtain sample tuples of inputs, outputs, and states $\left\{u_k^{(i)},x_k^{(i)},y_k^{(i)}\right\}_{i=1}^M$. \textcolor{black}{To obtain the observed states $z_k^{(i)}$, we simulate the discrete dynamics in \eqref{eq:KKL_discrete} using the unknown function $\omega(\cdot)$ for each $i=1,\dots,M$, over the time horizon $t\in[0,Nt_s]$, where $\omega(\cdot)$ will be approximated by a neural network as described in the following subsection.} The resulting data matrices are organized as follows:
\begin{equation*}\label{eq:Sample_Z}
Z=\begin{bmatrix}
    \left(z^{(i)}_0\right)^\top\\\left(z^{(i)}_1\right)^\top\\ \vdots \\ \left(z^{(i)}_{N-1}\right)^\top
    \end{bmatrix}\in\R^{N\times n_z}. \quad i=1,\dots,M,~k=0,\dots,N-1,
\end{equation*}

With the observer initialized and the data structured as above, we train a neural network to learn the functions $\omega(z|\theta)$ and $T^\dagger(z|\vartheta)$, as detailed in the following subsection.

\subsection{Neural state observation}
The chosen neural network architecture is a fully connected feedforward multilayer perceptron (MLP) network, selected for its simplicity and ease of implementation. An MLP is a composition of linear transformations $(p)$ and nonlinear activation functions $(\sigma(\cdot)$, e.g., ReLU, sigmoid, $\tanh$). The general form for a network with $D$ layers is as follows:
\begin{enumerate}
\item \textbf{Input layer ($\ell=0$):} The input to the network is the vector $z \in \mathbb{R}^{n_{z}}$, where $p_0=z$.
\item \textbf{Hidden layers ($\ell=1, 2, \dots, D-1$):} Each hidden layer applies the transformation:
\begin{equation*}
p_{\ell} = \sigma\left(W_\ell p_{\ell-1} + b_\ell \right) \in \mathbb{R}^{d_\ell},
\end{equation*}
where $W_\ell \in \mathbb{R}^{d_\ell \times d_{\ell-1}}$ is a weight matrix, $b_\ell \in \mathbb{R}^{d_\ell}$ is a bias vector.
\item \textbf{Output layer ($\ell=D$):} The final layer produces the output via a linear transformation:
\begin{equation*}
p_{D} = W_D h_{D-1} + b_D \in \mathbb{R}^{d_\mathrm{out}}.
\end{equation*}
\end{enumerate}

The overall network function $N\left(z|W_1,b_1,\dots,W_D,b_D\right)$, parameterized by $W_\ell$, is defined by the composition of all layer transformations:
\begin{equation*}\label{eq:Network_Comp}
\begin{aligned}
N\left(z|W_\ell\right) &= p_D \circ p_{D-1} \circ \cdots \circ p_1(z) \\
&= W_D \left( \sigma\left( W_{D-1} \left( \cdots \sigma\left( W_1 z + b_1 \right) \cdots \right) + b_{D-1} \right) \right) + b_D.
\end{aligned}
\end{equation*}
\noindent Here, $D$ denotes the total number of layers, and $d_\ell$ specifies the number of neurons (width) in the $\ell$-th hidden layer. The parameters $W_\ell \in \mathbb{R}^{d_\ell \times d_{\ell-1}}$ and $b_\ell \in \mathbb{R}^{d_\ell}$ are the weight matrix and bias vector for layer $\ell$, respectively, which are learned during training.

We seek to learn the parameterized function  \(\omega(z|\theta)\) to properly define the input dynamics, and the inverse map $T^\dagger(z|\vartheta)$ to obtain the state estimate. This parameterization is achieved using a two-network approach. To precisely describe the architecture of both networks, we introduce the following notation. The subscript $j$ denotes a specific network, where $j\in\left\{\omega,T^\dagger \right\}$. For each network $j$ let:
\begin{itemize}
    \item $D_j$ be its total number of layers,
    \item $d_{\ell,j}$ be the dimension (number of neurons) of its $\ell$-th hidden layer, and
    \item $d_{\text{out},j}$ be its output dimension.
\end{itemize}
\noindent The parameters $\theta$ and $\vartheta$ comprise the weights and bias for their respective networks. The hyperbolic tangent function, $\sigma(\cdot) = \tanh(\cdot)$, is chosen as the activation function for all hidden layers. Its bounded output prevents activations from blowing up, while its smooth, non-vanishing gradients allow for stable backpropagation to update the network weights and bias. We make the following \textcolor{black}{simplifying} assumption on the bias and weights of both networks for later analysis of generalization errors in Section~\ref{Sec: Gen_error}.
\begin{assumption}\label{Assum:Bias+Weight_Boundness}
    The weights and biases of both neural networks are bounded. Specifically, for network $\omega$ with $D_\omega$ layers, there exist constants $C_{W,\omega}, C_{b,\omega} > 0$ such that for all layers $\ell=1,\dots,D_\omega$, $\left\|W_{\ell,\omega}\right\|_{\mathrm{F}} \leq C_{W,\omega}$ and $\left\|b_{\ell,\omega}\right\|_2 \leq C_{b,\omega}$. Similarly, for network $T^\dagger$ with $D_{T^\dagger}$ layers, there exist corresponding bounds $C_{W,T^\dagger}, C_{b,T^\dagger}> 0$.
\end{assumption}

The parameters for  \(\omega(z|\theta)\) and \(T^\dagger(z)|\vartheta\)) are learned in a supervised fashion. 
Although our overall framework is model-free, the network training step uses the true system state $x_k$, which \textcolor{black}{is} assumed to be available during the offline training phase but \textcolor{black}{is} not available in online applications. The learning objective is to minimize the prediction error between the true state and the output of the inverse mapping network, leading to the following loss function (empirical risk):
\begin{equation}\label{eq:NN_Loss}
     \hat{\mathcal{L}}(\theta,\vartheta)=\frac{1}{MN}\sum_{i=1}^{M}\sum_{k=0}^{N-1}\left\|x_{k}^{(i)}-T^\dagger\left( z_k^{(i)}\big|\vartheta\right)\right\|_2^2
 \end{equation}
\noindent The loss function $\mathcal{L}(\theta,\vartheta)$ is minimized via gradient descent. The parameters of both networks are updated simultaneously using backpropagation to compute the gradients, as defined by the following update rule:
 \begin{equation}\label{eq:NN_backprop}
     \theta_{r+1}=\theta_r - \eta\nabla_{\theta_r}\mathcal{L}\left( \theta_r,\vartheta_r\right), \qquad  \vartheta_{r+1}=\vartheta_r - \eta\nabla_{\vartheta_r}\mathcal{L}\left( \theta_r,\vartheta_r\right),
 \end{equation}
where $r$ indexes the optimization step. Batches are constructed from individual trials, such that each batch corresponds to one full trajectory $i$. The networks are trained for $E$ epochs, where one epoch consists of $M$ optimization steps (one per training trajectory). We use the \texttt{RMSprop} optimizer to adapt the learning rate $\eta$ for each parameter, resulting in a total of $E \times M$ parameter updates. This results in one epoch representing a full pass through all $M$ training trajectories consisting of $N$ samples.

From the prior subsection, after collecting sample tuples of inputs, outputs, and states, we can initialize the observer by letting $z_0^{(i)}=0$. Once an initial estimate for the observed state $z$ is obtained, it is fed into the network $\omega$. The output of this network is then used to numerically solve the nonautonomous form \eqref{eq:KKL_discrete}, yielding a corrected observed state $z$. This corrected state is subsequently passed to the second network, $T^\dagger$, to obtain the final state estimate $\hat{x}$. The prediction error is then calculated via the loss function \eqref{eq:NN_Loss}, and the network parameters $(\theta,\vartheta)$ are updated via backpropagation. This entire training procedure is repeated for $E$ epochs. We refer to this approach as a \textbf{neural Luenberger observer for nonlinear systems with external inputs (NLOX)}. The entire procedure for state observation is outlined in Algorithm \ref{alg:one} and illustrated in Figure \ref{fig:NLOX_Diagram}.  

\begin{algorithm}[ht]
	\caption{Training procedure for NLOX}\label{alg:one}
	\KwIn{Tuple of system data $\left\{u_k^{(i)},y_k^{(i)},x_k^{(i)} \right\}$ for $k=0,1,\dots,N$ and $i=1,\dots,M$; Observer matrices $A$ and $B$; For $j\in \left\{\omega,T^\dagger \right\}$, number of layers $L_j$, number of neurons $d_{\ell,j}$; Initial neural network parameters $\theta_0$ and $\vartheta_0$; Number of epochs $E$}
	\KwOut{Learned mapping parameter $\vartheta$ and learned drift parameter $\theta$}
    \For{\text{epoch}$~=1,\dots E$}{
    \For{$i=1,\dots,M$}{Initialize  $z_0^{(i)}$\; 
    \For{$k=1,\dots,N$}{
    Solve \eqref{eq:KKL_discrete} to obtain $z_k^{(i)}$ under the given $\theta$\; } Evaluate network $T^\dagger$  for $\hat{x}_k^{(i)}$ under the given $\vartheta$\;
        Calculate $\hat{\mathcal{L}}(\theta,\vartheta)$ with \eqref{eq:NN_Loss}\;
    Backpropagate using \eqref{eq:NN_backprop} and update $(\theta,\vartheta)$\;
    }

    }
\end{algorithm}

\begin{figure}[ht]
\centering
\begin{tikzpicture}[
    node distance=2.2cm,
    box/.style={rectangle, draw, fill=blue!10, minimum width=2cm, minimum height=1cm, text centered},
    box2/.style={rectangle, draw, fill=green!10, minimum width=2cm, minimum height=1cm, text centered},
    box3/.style={rectangle, draw, fill=red!10, minimum width=2cm, minimum height=1cm, text centered},
    box4/.style={rectangle, draw,fill=white, minimum width=2.5cm, minimum height=2.0cm, text centered},
    trapleft/.style={trapezium, draw, fill=white, trapezium stretches=true,
                    trapezium left angle=70, trapezium right angle=70,
                    minimum width=1.5cm, minimum height=1.8cm, rotate=-90,
                    anchor=center},
    arrow/.style={->, >=stealth, thick},
    feedback/.style={-, >=stealth, thick, red, dashed},
    tune/.style={red, dashed, thick, ->, >=stealth}
]

\node[box, xshift = -0.5cm, yshift=1.5cm] (init) {Initialize $z_0$ and $\left(\theta,\vartheta\right)$};
\node[box4, above=of init,xshift = 0cm, yshift=-1.5cm] (sim) {};
\node[box2, right=of init, xshift=-1.3cm, yshift=0.0cm] (solve) {
    $\dot{z}=Az+By+\qquad\qquad\qquad u$
};
\node[above right=-1.95cm and -3.5cm of solve] (nn) {\includegraphics[width=3.5cm,height=2.5cm]{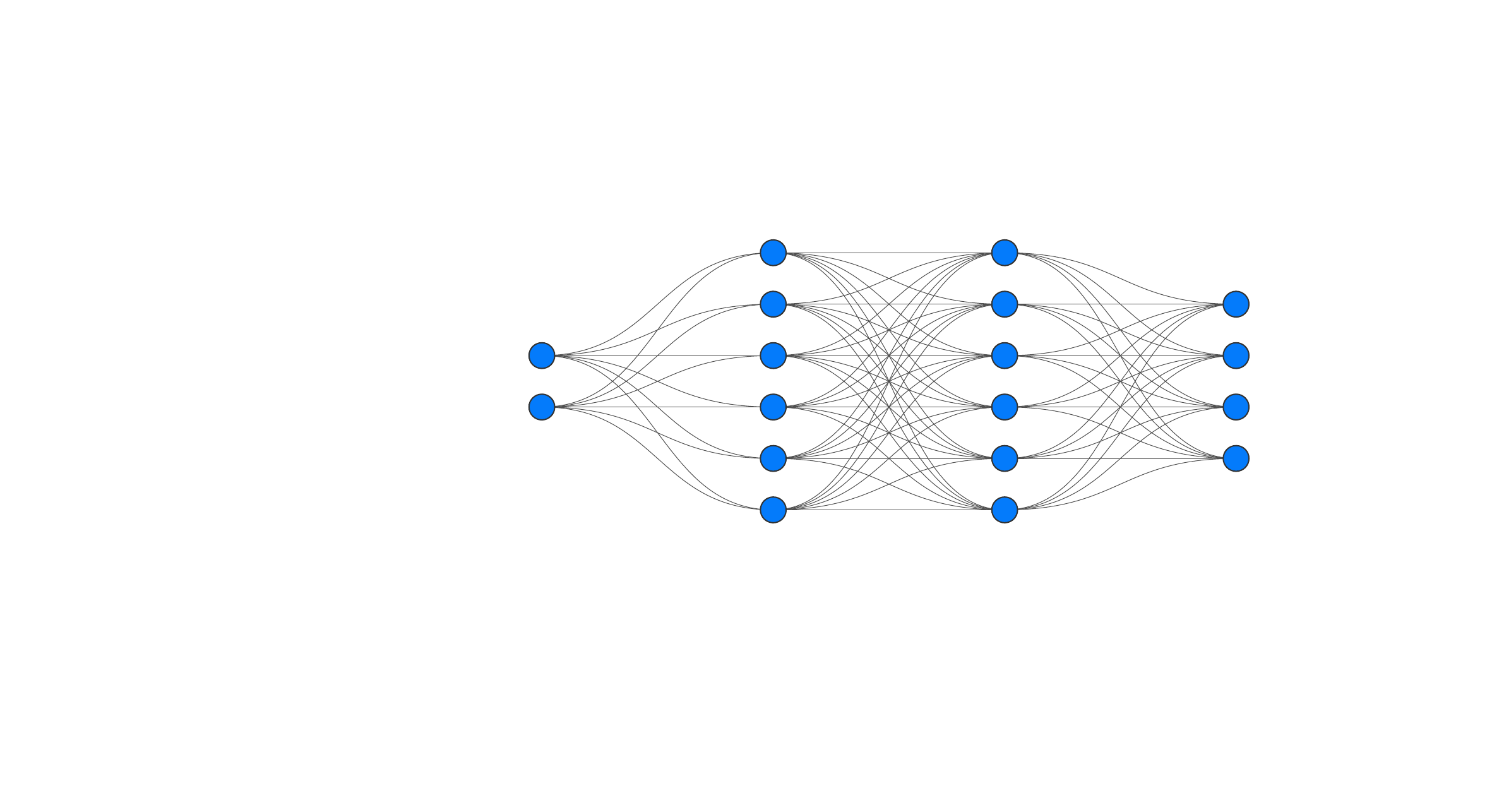}};

\node[trapleft, right=of solve, xshift=-0.58cm] (nn2) {\rotatebox{-90}{}};
\node[box3, right=of nn2, yshift=0.58cm] (error) {Calculate $\hat{\mathcal{L}}$
by \eqref{eq:NN_Loss}};

\node at ([xshift=-3mm,yshift=-1.1mm]nn2.center) {\includegraphics[width=3.5cm,height=2.5cm]{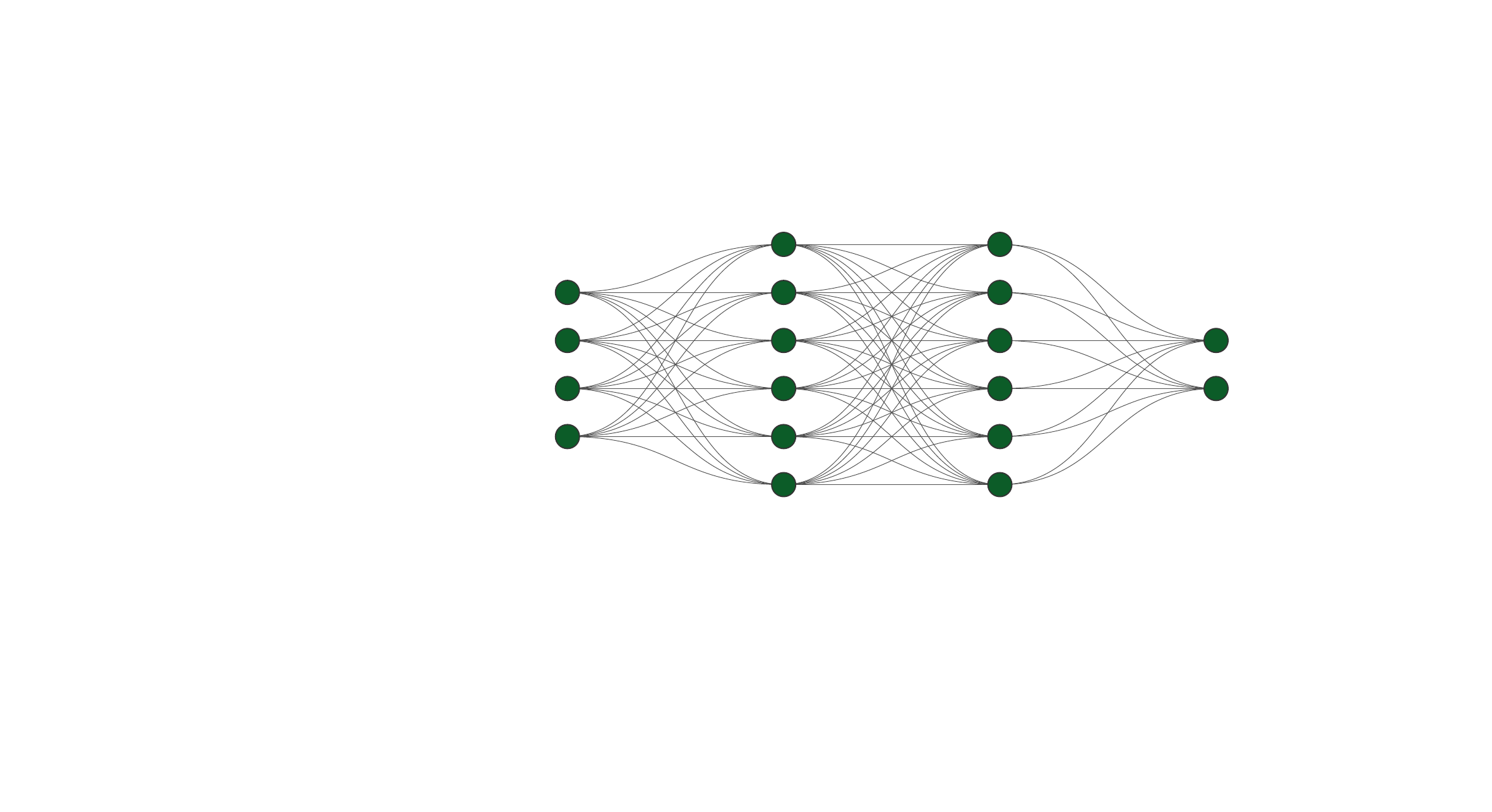}};
\node at (sim.center) {\includegraphics[width=2cm,height=2cm]{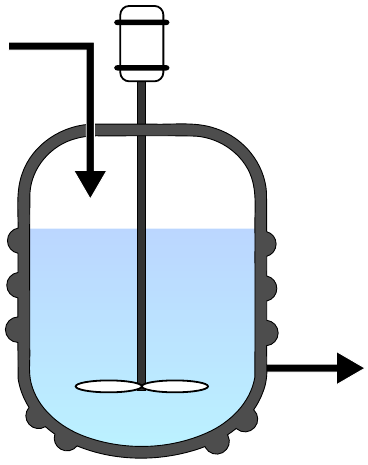}};

\draw[arrow] (init.east) -- node[midway, above] {}(solve.west);
\draw[arrow] (solve.east) -- node[midway, above] {$z$}(nn2.south);
\draw[arrow] (nn2.north) -- node[midway, above] {$\hat{x}$}(error.west);

\draw[arrow] (sim.east) -- ++(0,0) -- ++(0,0) -|  node[pos=0.25, above] {$(u,y)$} (solve.north);

\draw[feedback] (error.south) -- ++(0,0) -- ++(0,-0.6) -| ([xshift=0.9cm]solve.south);
\draw[feedback] (error.south) -- ++(0,0) -- ++(0,-0.3) -| (nn2.east);

\begin{scope}[on background layer]
    \draw[tune] ([xshift=0.9cm]solve.north) -- ++(65:0.65cm) node[above] {};
    \draw[tune] ([xshift=-.4cm]nn2.east) -- ++(65:1.85cm) node[above] {};
\end{scope}

\end{tikzpicture}
\caption{Neural Leunberger state observer for systems with exogenous inputs.}
\label{fig:NLOX_Diagram}
\end{figure}
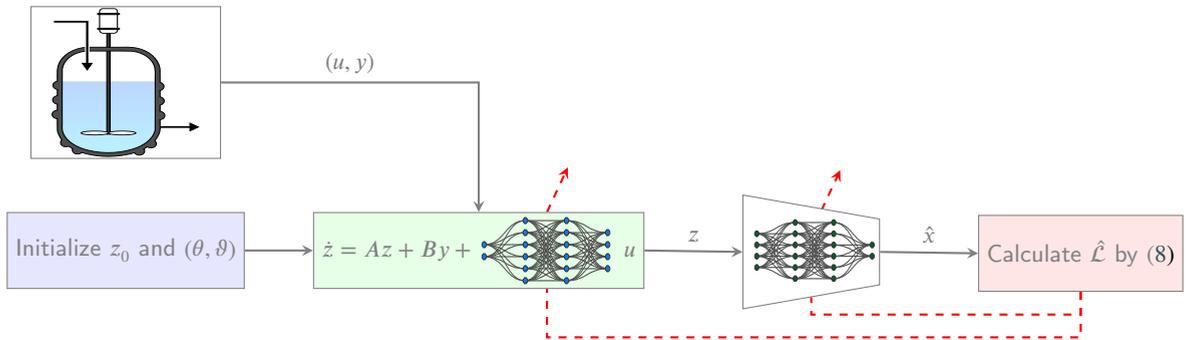

\SetNlSty{}{}{:} 
\SetAlgoNlRelativeSize{-1} 
\SetAlgoLined 
\DontPrintSemicolon 
\textcolor{black}{\remark{Assumption \ref{Assum:Bias+Weight_Boundness} ensures that the learned observer components from the neural networks remain bounded and well behaved. In practice, the boundedness of the neural network weights and biases is enforced implicitly through stable training procedures, such as early stopping, bounded activation functions (e.g., $\tanh$), and appropriate initialization. While exact parameter bounds are not verified analytically from data, monitoring the magnitude of network parameters during training provides a practical means of ensuring stable learning behavior. When this assumption is only approximately satisfied, the resulting estimation error bounds may become looser, but the observer remains well posed in practice.}}
\section{Generalization error of the observer}\label{Sec: Gen_error}
After training is completed, the NLOX observer is evaluated on unseen test data. To formally quantify its generalized performance to observe the states, we establish a theoretical, probabilistic upper bound on its observation error. 
\subsection{Problem formulation}
We begin by formalizing the objective function \eqref{eq:NN_Loss}, which is computed over a finite data set of $M$ trajectories, each of length $N$, where for each trajectory $i$, a loss is assigned depending on the random initial state $x_0^{(i)}$ and the random input sequence $\left(u_0^{(i)},\dots u_{N-1}^{(i)}\right)$. This  yields the empirical risk over $M$ trajectories as:  
\begin{equation*}\label{eq: EmpiricalRisk1}
    \hat{\mathcal{R}}_M = \frac{1}{MN}\sum_{i=1}^{M}\sum_{k=0}^{N-1}\left\|x_{k}^{(i)}-T^\dagger\left( z_k^{(i)}\big|\vartheta\right) \right\|^2.
\end{equation*}
The empirical risk defined above is a sample average $\frac{1}{M}\sum_{i=1}^M \left\|\Upsilon_i\right\|^2$  of the true \emph{risk} as an expectation over the entire population $\mathscr{D}$:
\begin{equation*}\label{eq:risk}
    \mathcal{R} = \mathbb{E}_{\mathscr{D}} \left[\frac{1}{N}\sum_{k=0}^{N-1}\left\|x_k-T^\dagger\left(z^\star\left(kt_s  \right)\big|\vartheta \right)\right\|^2\right]. 
\end{equation*}
Due to the use of discretization as discussed in Subsection~\ref{sec:data_sampling}, since $z^\star$ is not known but only approximated by $z_k$ with error $\varepsilon_k$, we should consider
\begin{equation*}
        \hat{\mathcal{R}}_\infty = \mathbb{E}_{\mathscr{D}} \left[\frac{1}{N}\sum_{k=0}^{N-1}\left\|x_k-T^\dagger\left(\textcolor{black}{z_k} \big|\vartheta \right)\right\|^2\right]. 
\end{equation*}
For our analysis, we express $\hat{\mathcal{R}}_M$ as an average over per-sample components, for each sample $i=1,\dots,M$, we have
\begin{equation*}\label{Per_Sample_EmpiricalRisk}
\hat{\mathcal{R}}_M =\frac{1}{M}\sum_{i=1}^M\left\|\Upsilon^{(i)} \right \|^2, \quad \text{where } 
    \Upsilon^{(i)} = \frac{1}{\sqrt{N}}\begin{bmatrix}
        x_{1}^{(i)} - T^\dagger\left(z_0^{(i)}\big|\vartheta\right)\\
        x_{2}^{(i)} - T^\dagger\left(z_1^{(i)}\big|\vartheta\right)\\ \vdots\\
        x_{N}^{(i)} - T^\dagger\left(z_{N-1}^{(i)} \big|\vartheta\right)
    \end{bmatrix}.
\end{equation*}
We now view $\Upsilon$ as a random vector whose value depends on the random initial state $x_0$ and the input sequence $\left(u_0, \dots, u_{N-1}\right)$, which are jointly distributed according to $(x_0, u_0, \dots, u_{N-1}) \sim \mathscr{D}$, i.e.,
\begin{equation}\label{eq:Upsilon}
\Upsilon = \Upsilon(x_0, u_0, \dots, u_{N-1}) := \frac{1}{\sqrt{N}}\begin{bmatrix}
        x_{1} - T^\dagger\left( z_0  \big|\vartheta\right)\\
        x_{2} - T^\dagger\left(z_1 \big|\vartheta\right)\\ \vdots\\
        x_{N} - T^\dagger\left(z_{N-1} \big|\vartheta\right)
    \end{bmatrix}.
\end{equation}
Given $M$ independent samples, the empirical joint distribution of these random variables is given by 
\begin{equation*}\label{X_emp_dist}
\mathscr{S} = \frac{1}{M} \sum_{i=1}^M \delta_{\left( x_0^{(i)}, u_0^{(i)}, \dots, u_{N-1}^{(i)} \right)}, 
\end{equation*}
where $\delta$ is the single-point probability measure. Therefore, the vectors $\Upsilon^{(1)}, \dots, \Upsilon^{(M)}$ are independent realizations of the random vector $\Upsilon$. Therefore, we aim to bound the difference between the empirical risk and its expected value in its true distribution $\mathscr{D}$.

\subsection{Lipschitz continuity}
The generalization analysis relies on establishing the Lipschitz continuity of the loss with respect to the observed state $z$. We begin our analysis with a further assumption on the initial state distribution.
\begin{assumption}\label{Assum:boundedness_x0_u}
    The distribution of the initial states $x_0$ is supported on a compact set with $\left\|x_0\right\|_\infty \leq X_*$ for some constant $X_* > 0$ (i.e., $\mathcal{X}_0$ is bounded). Its components are independent, and $x_0$ is independent of the input sequence $u_0, \dots, u_{N-1}$. Moreover, to use Fact~\ref{fact:Bernard_NonAuto}, the input $u(t)$ is designed to be uniformly bounded in amplitude. That is, there exists a constant $U_*>0$ such that $\|u(t)\| \le U_*,~ \forall t \ge 0$ (i.e., $\mathcal{U}$ is bounded). \textcolor{black}{In addition, $t_s$ is small enough such that $I+t_sA$ is Schur stable.}
\end{assumption} 
Recall that the random vector $\Upsilon$ defined above depends on the rollout $\left( x_k, y_k,u_k,z_k\right)$, which in turn is determined by the initial condition $x_0$ and the applied input sequence $\left(u_0,\dots,u_{N-1}\right)$. In order to obtain a generalization bound, we will show that $\Upsilon$ is Lipschitz with respect to the underlying random variables $\left( x_0, u_0, \dots, u_{N-1}\right)$. To make this precise, we first formalize boundedness properties of the data we draw.

\begin{lemma}\label{Lemma:bounded_outputs}
Under Assumption~\ref{Assum:Solvability}, there exists a constant $Y_*>0$ such that
$$
\|y(t)\| = \|h\left(x(t)\right)\| \le Y_*, \qquad \forall t \ge 0.
$$
\end{lemma}
\noindent The lemma can be easily verified by the fact that the state trajectory $x(t)$ remains in the compact set $\mathcal{X}$ for all $t \ge 0$. Since $h:\mathcal{X}\to\mathbb{R}^{n_y}$ is smooth, and therefore continuous, its image $h(\mathcal{X})$ is also compact. The same bounds hold in discrete time by denoting $y_k := y(k t_s)$ and $u_k := u(k t_s)$ as the sampled signals:
$$
\left\|y_k\right\| \leq Y_*, \qquad \left\|u_k\right\| \leq U_*, \qquad \forall k = 0,\dots,N-1.
$$
Since the observer \eqref{eq:KKL_discrete} depends only on the sampled input-–output pair $\left( u_k,y_k\right)$, Assumption~\ref{Assum:boundedness_x0_u} and Lemma~\ref{Lemma:bounded_outputs} imply that the observer is driven by uniformly bounded signals. In particular, together with the boundedness of the initialization, this guarantees that the observer state $z_k$ remains uniformly bounded over the entire horizon.

\begin{lemma}\label{Lemma:Z_Bound}
  Suppose \eqref{eq: Input-Affine Dynamic System} satisfies Assumptions \ref{Assum:Solvability}, \ref{Assum:deltaISS}, \ref{Assum:Bias+Weight_Boundness}, and \ref{Assum:boundedness_x0_u}. Then $\exists Z_*>0$ such that for all $i=1,\dots,M,$ as long as $\left\|z_0^{(i)} \right\|\leq Z_*$,
  \begin{equation}
    \left\|z_k^{(i)}\right\| \leq Z_*,~ \forall ~k=1,\dots,N-1,.
\end{equation}
\end{lemma}
\begin{proof}
    \textcolor{black}{Since $A$ is Hurwitz by construction, there exists a $P \succ 0$ such that $A^\top P+PA=-Q$ for some $Q\succ0$. For sufficiently small sampling time $t_s>0$, the matrix $I+t_sA$ is Schur stable, i.e., there exists a $P \succ 0$ such that $\left(I+t_sA \right)^\top P\left(I+t_sA \right) -P= -t_sI \prec 0$.} Consider the Lyapunov function $V\left(z_k\right) = z_k^\top P z_k$. \textcolor{black}{The difference of $V$ along the trajectories is:}
\begin{equation*}\begin{aligned}\label{eq:Lyapunov1}
        {V}\left(z_{k+1}\right)-{V}\left(z_{k}\right)&=\left[ \left(I+t_sA\right)z_k+t_s\left(By_k+\omega\left(z_k\right)u_k\right)\right]^\top P~[\cdot]-z_k^\top P z_k\\
        &=-z_k^\top \left(t_s I\right)z_k+2t_s\left(By_k+\omega\left(z_k\right)u_k \right)^\top P \left(I+t_sA \right)z_k+t_s^2\left(By_k+\omega\left(z_k\right)u_k \right)^\top P (\cdot).
\end{aligned}
\end{equation*}
We analyze the last two terms on the right-hand side. For the second term, by Cauchy-Schwarz inequality, we have $\left|2t_s \left(By_k+\omega\left(z_k\right)u_k \right)^\top P \left(I+t_sA \right)z_k\right|\leq2t_s \left(\|B\|Y_*+M_\omega U_* \right)\|P\|\left\|I+t_sA\right\| \left\|z_k \right\|$, since $\omega$ is bounded by some $M_\omega>0$ by Assumption \ref{Assum:Bias+Weight_Boundness} (due to bounded weights/biases and the 1-Lipschitz $\tanh$ activation). For the last term, we have $t_s^2\left| \left(By_k+\omega\left(z_k\right)u_k \right)^\top P (\cdot)\right| \leq t_s^2\|P\|\left(\|B\|Y_*+M_\omega U_* \right)^2$. Putting all the bounds together, we obtain:
\begin{equation*}\label{Lyapunov 2}
\begin{aligned}
{V}\left(z_{k+1}\right)-{V}\left(z_{k}\right) \leq - t_s\left\|z_k\right\|^2+ 2t_sL_1 \left\|z_k\right\| + t_s^2L_2^2 .
\end{aligned}
\end{equation*}
From this inequality, it follows that $V\left(z_{k+1}\right) \leq {V}\left(z_k\right)$ whenever $\|z_k\| \geq L_1+\sqrt{L_1^2+t_s^2L_2^2}$, proving that $z$ is ultimately bounded by this quantity. The theorem's conclusion thus follows.
\end{proof}

This establishes that $z$ remains bounded for bounded $(u,y)$, provided the stability of $A$. This further allows us to conclude that, in combination with Assumption \ref{Assum:Bias+Weight_Boundness}, the $T^\dagger$ network is $L_{T^\dagger}$-Lipschitz and the $\omega$ network is $L_{\omega}$-Lipschitz on the relevant compact domain. Now consider the $k$-th component of $\Upsilon$ defined in \eqref{eq:Upsilon}. 
\begin{equation*}
     \sqrt{N}\Upsilon_k=x_{k} - T^\dagger\left(z_{k-1}+t_s\left( Az_{k-1} +By_{k-1} +\omega\left(z_{k-1}\big|\theta \right)u_{k-1}^{(i)} \right)\big|\vartheta\right).
\end{equation*}
The observer dynamics \eqref{eq:KKL-observer-NonAuto} depend on the initial state condition  $x_0$ and input sequence $u_0, \dots, u_{N-1}$, but inherit several structural properties: (i) the matrix $A$ is Hurwitz by design, ensuring exponential stability; (ii) the learned term $\omega(\cdot|\theta)$ is Lipschitz, by Assumption~\ref{Assum:Bias+Weight_Boundness}; (iii) the input $u(t)$ is uniformly bounded, and the output $y(t)=h(x(t))$ is uniformly bounded and Lipshitz in $x$ by Assumption~\ref{Assum:boundedness_x0_u} and Lemma~\ref{Lemma:bounded_outputs}; and (iv) differences in $x(t)$ over two trajectories are already controlled by Assumption~\ref{Assum:deltaISS}. Together, these properties imply that the observer state $z_k$ is incrementally input-to-state stable with respect to $\left(x_0, u_0, \dots, u_{N-1}\right)$. We formalize this via the following incremental stability condition for the observed states.

\begin{lemma}\label{lemma:deltaISS_observer}
Consider two trajectories of the observer dynamics \eqref{eq:KKL-observer-NonAuto} driven by two input trajectories \(\big(x_0,u(\cdot)\big)\) and \(\big(x_0',u'(\cdot)\big)\), and initialized with the same observer state \(z(0)=z'(0)=0\).
There exist a class-\(\mathscr{KL}\) function \(\beta_z\) and a class-\(\mathscr{K}\) function \(\gamma_z\) such that, for all \(t \ge 0\),
$$
\| z(t) - z'(t) \|
\;\le\;
\beta_z\!\big(\,\|x_0 - x_0'\|,\, t \big)
\;+\;
\gamma_z\!\Big(
\operatorname*{ess\,sup}_{0 \leq \tau \leq t} \| u(\tau) - u'(\tau) \|
\Big).
$$
Moreover, on the compact domain reached under Assumptions~\ref{Assum:Solvability}--\ref{Assum:boundedness_x0_u}, \(\beta_z\) and \(\gamma_z\) admit linear upper bounds of the form
$$
\beta_z(r,t) \le L_{\beta_z,t}\, r,
\qquad
\gamma_z(r) \le L_{\gamma_z}\, r,
$$
for some nonnegative gains \(L_{\beta_z,t}\) and \(L_{\gamma_z}\) that do not depend on the particular realization of \((x_0,u(\cdot))\).
\end{lemma}

To analyze how $\Upsilon$ depends on the initial condition $x_0$ and input sequence $u_0, \dots, u_{N-1}$, we exploit the $\delta$ISS assumptions.
\begin{lemma}\label{Lemma_Lipschitz}
    Under Assumptions \ref{Assum:Solvability}, \ref{Assum:deltaISS}, \ref{Assum:Bias+Weight_Boundness}, and \ref{Assum:boundedness_x0_u} the mapping $\left(x_0, u_0, \dots, u_{N-1}\right) \mapsto \Upsilon$ is Lipschitz continuous with constant:
    $$L_\Upsilon=\max\left(\tilde{L}, \tilde{L}_0, \dots, \tilde{L}_{N-1}\right) , $$
    with
$$\tilde{L}:=\frac{X_* L_{\beta,t_s}\left(1- L_{\beta,t_s}^N\right)}{1- L_{\beta,t_s}}+\frac{ L_{T^\dagger}L_{\beta_z,t_s}\left(1- L_{\beta_z,t_s}^N\right)}{1- L_{\beta_z,t_s}}\leq\frac{X_* L_{\beta,t_s}}{1- L_{\beta,t_s}}+\frac{ L_{T^\dagger}L_{\beta_z,t_s}}{1- L_{\beta_z,t_s}},$$
    and for $k=0,1,\dots,N-1$
    $$\tilde{L}_k:=L_\gamma \frac{1- L_{\beta,t_s}^{N-k}}{1- L_{\beta,t_s}}+L_{T^\dagger}L_{\gamma_z} \frac{1- L_{\beta_z,t_s}^{N-k}}{1- L_{\beta_z,t_s}}\leq  \frac{L_\gamma}{1- L_{\beta,t_s}}+\frac{L_{T^\dagger}L_{\gamma_z}}{1- L_{\beta_z,t_s}}.$$
\end{lemma}
\begin{proof}
    Using Assumption \ref{eq:deltaISS} on incremental input-to-state stability, for the system evolving from $(k-1)t_s$ to $kt_s$ ($k\geq 1$), we have
    $$\left\|x_k - x_k'\right\| \leq L_{\beta,t_s} \left\|x_{k-1}-x_{k-1}'\right\| + L_\gamma \max_{\tau\in \left[(k-1)t_s, kt_s\right]} \left\|u(\tau)-u'(\tau)\right\|, $$
    Then by backwards induction on $k$, 
    $$ \left\|x_k - x_k'\right\| \leq X_*L_{\beta,t_s}^k \left\|x_0-x_0'\right\| + L_\gamma \sum_{j=0}^{k-1}L_{\beta, t_s}^{k-1-j}\left\|u_j - u_j'\right\|.$$
Given when $\left\|x_0\right\|_\infty \leq X_*$ is satisfied, for the resulting $\Upsilon_k$ and $\Upsilon_k'$ from the two trajectories, we have

$$ \sqrt{N} \left\|\Upsilon_k - \Upsilon_k'\right\| \leq\left\|x_k - x_k' - \left( T^\dagger \left(z_{k}\right)-T^\dagger\left(z'_{k}\right)\right)\right\|\leq\left\|x_k - x_k' \right\| +L_{T^\dagger} \left\| z_{k}-z'_{k}\right\| ,$$
which simplifies to 
    $$ \sqrt{N} \left\|\Upsilon_k - \Upsilon_k'\right\| \leq X_*L_{\beta,t_s}^k \left\|x_0-x_0'\right\| + L_\gamma \sum_{j=0}^{k-1}L_{\beta, t_s}^{k-1-j}\left\|u_j - u_j'\right\| +L_{T^\dagger} \left\| z_{k}-z'_{k}\right\| .$$
 
\noindent  By Lemma~\ref{lemma:deltaISS_observer}, the observer dynamics are incremental input-to-state stability; therefore, it admits an analogous bound in terms of $x_0$ and $u_j$ through backwards induction on $k$. 
 $$
 \sqrt{N} \left\|\Upsilon_k - \Upsilon_k'\right\| \leq \left( X_*L_{\beta,t_s}^k+L_{\beta_z,t_s}^{k}L_{T^\dagger}\right) \left\|x_0-x_0'\right\| + \sum_{j=0}^{k-1}\left(L_\gamma L_{\beta, t_s}^{k-1-j}+L_{\gamma_z}L_{T^\dagger}L_{\beta_z, t_s}^{k-1-j} \right)\left\|u_j - u_j'\right\| .
 $$

\noindent Then by summing over $k=1,\dots,N$ and applying the geometric-series identity, we obtain 
\begin{equation*}\label{eq:Upsilon_sum_bound}
\begin{aligned}
\sqrt{N}\,\|\Upsilon - \Upsilon'\|
\;\le\;&
\left[
X_* \,L_{\beta,t_s}\,\frac{1 - L_{\beta,t_s}^{N}}{1 - L_{\beta,t_s}}
+
L_{T^\dagger}\,L_{\beta_z,t_s}\,\frac{1 - L_{\beta_z,t_s}^{N}}{1 - L_{\beta_z,t_s}}
\right]
\left\|x_0 - x_0'\right\| \\[0.5em]
&+
\sum_{j=0}^{N-1}
\left[
L_\gamma \,\frac{1 - L_{\beta,t_s}^{\,N-j}}{1 - L_{\beta,t_s}}
+
L_{T^\dagger} L_{\gamma_z} \,\frac{1 - L_{\beta_z,t_s}^{\,N-j}}{1 - L_{\beta_z,t_s}}
\right]
\left\|u_j - u_j'\right\|.
\end{aligned}
\end{equation*}

\noindent This implies that the random vector $\Upsilon$ is Lipschitz on $\left(x_0, u_0, \dots, u_{N-1}\right)$ with Lipschitz constant 
\begin{equation*}\label{eq:L_tilde_defs}
\tilde{L}
:=
X_* \,L_{\beta,t_s}\,\frac{1 - L_{\beta,t_s}^{N}}{1 - L_{\beta,t_s}}
+
L_{T^\dagger}\,L_{\beta_z,t_s}\,\frac{1 - L_{\beta_z,t_s}^{N}}{1 - L_{\beta_z,t_s}},
\qquad
\tilde{L}_j
:=
L_\gamma \,\frac{1 - L_{\beta,t_s}^{\,N-j}}{1 - L_{\beta,t_s}}
+
L_{T^\dagger} L_{\gamma_z} \,\frac{1 - L_{\beta_z,t_s}^{\,N-j}}{1 - L_{\beta_z,t_s}},
\end{equation*}
    \noindent 
    $$L_\Upsilon=\max\left(\tilde{L}, \tilde{L}_0, \dots, \tilde{L}_{N-1}\right) . $$
    The lemma's conclusion thus follows. 
\end{proof}
\textcolor{black}{Lemma~\ref{Lemma_Lipschitz} establishes that the mapping from sampled initial conditions and input sequences to  $\Upsilon$ is Lipschitz continuous. This implies that the generalization error depends smoothly on the underlying data and does not exhibit uncontrolled sensitivity to perturbations. This regularity property is important in the subsequent derivation of finite-sample generalization bounds.}

\subsection{Rademacher complexity and concentration inequality}\label{subsec: Rade_concen_inequal}
Since the observed dynamics incorporate learned components $(\omega, T^\dagger)$ that depend on both the observed states $z_k$ and the underlying true states $x_k$. Consequently, we define the model class as a composition of function classes. Specifically, the inverse map from $z_k$ to $x_k$ is given by $T^\dagger\left(z_k \ \big| \ \vartheta \right)$, representing the composition of the $T^\dagger$-network with a function whose input depends on the $\omega$-network. We now formally define the model classes for both neural networks:
\begin{equation*}
\begin{aligned}
    \mathcal{H}_{\omega} &= \left\{ \omega(\cdot | \theta): z \mapsto  W_{D_\omega} \left( \sigma\left( W_{{D_\omega}-1} \left( \cdots \sigma\left( W_1 z + b_1 \right) \cdots \right) + b_{{D_\omega}-1} \right) \right) + b_{D_\omega}\ \big| \ \theta = \left(W_1, b_1, \dots, W_{D_\omega}, b_{D_\omega}\right) \right\},\\  
\mathcal{H}_{T^\dagger}& = \left\{ T^\dagger(\cdot | \vartheta): z \mapsto  W_{D_{T^\dagger}} \left( \sigma\left( W_{{D_{T^\dagger}}-1} \left( \cdots \sigma\left( W_1 z + b_1 \right) \cdots \right) + b_{{D_{T^\dagger}}-1} \right) \right) + b_{D_{T^\dagger}}\ \big| ~\vartheta = \left(W_1, b_1, \dots, W_{D_{T^\dagger}}, b_{D_{T^\dagger}}\right) \right\}.
\end{aligned}
\end{equation*} 
Let $\mathcal{F}$ be the overall function class defined as:
    \begin{equation*}
\mathcal{F} = \left\{ f: \left(x_0,u_0,\dots,u_{N-1}\right) \mapsto T^\dagger\left(  z_k  \ \big| \ \vartheta \right) \ \middle| \ \omega \in \mathcal{H}_{\omega},\ T^\dagger \in \mathcal{H}_{T^\dagger} \right\}.
\end{equation*}

To derive the generalization bound, we rely on Rademacher complexity theory \citep{shalev2014understanding}, where we first define the empirical Rademacher complexity with respect to each $i$-th sample:
\begin{equation*}
\text{Rad}(\mathcal{F}) = \mathbb{E}_\varepsilon \left[\sup_{f\in\mathcal{F}}\frac{1}{M}\sum_{i=1}^{M}\varepsilon_i f\left(x_0^{(i)},u_0^{(i)},\dots,u_{N-1}^{(i)}\right) \right],
\end{equation*}
where $\varepsilon_i$ are independent Rademacher random variables (i.e., $\varepsilon_i\in{-1,+1}$ with equal probability). These variables serve to ``randomly flip'' the contributions of each data point in the empirical loss, allowing us to quantify the complexity of the model classes $\mathcal{H}_\omega,~\mathcal{H}_{T^\dagger}$. However, our generalization bound requires the Rademacher complexity of the loss-composed class $\xi \circ \mathcal{F}$ where $\xi = |\Upsilon|^2$. Therefore, we rely on the following fact.

\begin{fact}[\cite{ledoux2013probability}\label{Fact:Contract}, Corollary 3.17]
Given a $L_\xi$-Lipschitz mapping $\xi:\R\to\R$, the Rademacher complexity of the composed class satisfies:
$$\text{Rad}(\xi\circ\mathcal{F})\leq2L_\xi ~\text{Rad}(\mathcal{F}).$$
\end{fact}
\noindent Readers are referred to Section 26 in \cite{shalev2014understanding}. In our case, we consider $\xi:\Upsilon\mapsto|\Upsilon|^2$, which is $2\Upsilon^*$-Lipschitz, where $\Upsilon^*$ is an upper bound on the magnitude of $\Upsilon$ due to the boundness of all its arguments and mappings in its definition in \eqref{eq:Upsilon}. Hence
$$\text{Rad}(\xi\circ\mathcal{F})\leq2\Upsilon^*L_\Upsilon ~\text{Rad}(\mathcal{F}),$$
where $L_\Upsilon$ was established in Lemma~\ref{Lemma_Lipschitz}.
\begin{theorem}
    
Under the prior made assumptions, there exists a constant $c>0$, such that with probability of at least $1-\delta$,
    \begin{equation}
\left| \mathbb{E}_{\mathscr{D}}\left[|\Upsilon|^2\right] - \frac{1}{M}\sum_{i=1}^M \left|\Upsilon^{(i)}\right|^2 \right| \leq 2~\text{Rad}(\mathcal{\xi \circ F}) +c\sqrt{\frac{2\log(2/\delta)}{MN}}.
\end{equation}
\end{theorem}

\noindent Then by applying Fact \ref{Fact:Contract} we obtain
\begin{equation}
    \left| \mathbb{E}_{\mathscr{D}}\left[|\Upsilon|^2\right] - \frac{1}{M}\sum_{i=1}^M \left|\Upsilon^{(i)}\right|^2 \right| \leq 4\Upsilon^*L_\Upsilon~\text{Rad}(\mathcal{ F}) +c\sqrt{\frac{2\log(2/\delta)}{MN}},
\end{equation}

To complete the bound, we analyze $\text{Rad}(\mathcal{F})$. Since $\mathcal{F}$ consists of neural networks with bounded weights/biases and inputs, the Rademacher complexity can be expressed as in the following Lemma. 
\begin{lemma}\label{Lem:Radmacher_Complex}
        Let $\mathcal{F}$ be the class of neural networks described above, with weight matrices and bias vectors bounded by $C_{W,\omega}, C_{b,\omega}, C_{W,T^\dagger}, C_{b,T^\dagger}$, and evaluated on a dataset of size $MN$. Then the empirical Rademacher complexity of $\mathcal{F}$ satisfies
    \begin{equation}
        \textnormal{Rad}(\mathcal{F}) \leq \frac{C_{W,\omega} C_{b,\omega} C_{W,T^\dagger} C_{b,T^\dagger}}{\sqrt{MN}}.
    \end{equation}
\end{lemma}

\begin{corollary}\label{FinalBound}
    Under all assumptions mentioned above, with confidence at least \(1 - \textcolor{black}{\delta}\) for any $\textcolor{black}{\delta}\in(0,1)$:
\begin{equation}
    \begin{aligned}
        \left|\hat{\mathcal{R}}_\infty -\hat{\mathcal{R}}_M  \right| \leq \frac{4\Upsilon^*L_\Upsilon C_{W,\omega}C_{b,\omega}C_{W,T^\dagger}\ C_{b,T^\dagger}}{\sqrt{MN}} +c\sqrt{\frac{2\log(2/\delta)}{MN}},
    \end{aligned}
\end{equation}
where the Lipschitz constant $L_\Upsilon$ is given in Lemma \ref{Lemma_Lipschitz}. 
\end{corollary}

\noindent As expected, the difference between empirical risk and the actual risk over the entire population scales as $\mathcal{O}\left(1/\sqrt{MN}\right)$, where $MN$ represents the total number of state observations across all trajectories. \textcolor{black}{This scaling explicitly characterizes the finite-sample behavior of the learning procedure and shows that the proposed observer is statistically consistent: as either the number of trajectories or the trajectory length increases, the empirical risk converges to the expected risk.}

\subsection{Forward Euler discretization error}\label{Sec: Trunc_error}
In our setting, the input $u$ is designed as a staircase (piecewise-constant) signal, while the output $y=h(x)$ is smooth. As a result, the right-hand side of \eqref{eq:KKL-observer-NonAuto} is only piecewise smooth in time. A first-order (forward Euler) discretization across these pieces introduces a local truncation error at each step, which we collect in the term $\varepsilon_k$. Consider the discrete KKL observer from \eqref{eq:KKL_discrete} where $\varepsilon_k$ denotes the one–step truncation error of the forward Euler scheme applied to the continuous observer. Let $z^\star(t)$ denote the exact continuous-time observer trajectory satisfying \eqref{eq:KKL-observer-NonAuto} with the same initial condition, and write $z^\star(k t_s)$ for its sample at time $k t_s$. Under the standing Lipschitz assumptions (bounded $u_k,y_k$, Lipschitz $\omega$, and Hurwitz $ A$), it is a well-known fact from numerical ODE theory that:
\begin{equation}\label{eq:global_err_bound_simple}
   \max_{0\le k\le N}\,\left\|\varepsilon_k\right\|=\max_{0\le k\le N}\,\big\|z_k - z^\star\left(k t_s\right)\big\|\leq L_{\mathrm{disc}}\, t_s.  
\end{equation}

\noindent for a constant $L_{\mathrm{disc}}>0$ that depends only on boundedness and smoothness on the compact domain (not on $M,N$).
Consequently, at each step the map $T^\dagger(\cdot\mid\vartheta)$ incurs at most
\begin{equation}\label{eq:decode_err_FE}
\big\|T^\dagger\left(\,\cdot+\varepsilon_k\,\mid\vartheta\right)-T^\dagger(\,\cdot\,\mid\vartheta)\big\|
\;\le\;
L_{T^\dagger}\,L_{\mathrm{disc}}\,t_s, 
\end{equation}
namely, $\big\|T^\dagger\left(z_k\mid\vartheta\right)-T^\dagger\left(z^\star\left( kt_s\right)\mid\vartheta\right)\big\|
\;\le\;
L_{T^\dagger}\,L_{\mathrm{disc}}\,t_s$. Thus, the forward Euler approximation contributes a small,  additive bias (per step over a fixed horizon $O\left(t_s\right)$), which can be made arbitrarily small by choosing $t_s$ sufficiently small. 

Returning to the definition of risk $\mathcal{R}$, we see that $\left\|\mathcal{R}-\hat{\mathcal{R}}_\infty \right\|\leq L_{T^\dagger}\,L_{\mathrm{disc}}\,t_s.$ Combining this discretization error with the statistical bound from Corollary~\ref{FinalBound} yields the following overall risk estimate.

\begin{corollary}\label{Col:true_Final_Bound}
    Under all assumptions stated above, with confidence at least \(1 - \delta\) for any \(\delta \in (0,1)\):
    \begin{equation}\label{FinalTrueBound}
            \left|\mathcal{R} - \hat{\mathcal{R}}_M \right|
            \leq  \frac{4\Upsilon^* L_\Upsilon\, C_{W,\omega} C_{b,\omega} C_{W,T^\dagger} C_{b,T^\dagger}}{\sqrt{MN}}
            + c\sqrt{\frac{2\log(2/\delta)}{MN}}+L_{T^\dagger} L_{\mathrm{disc}}\, t_s,
    \end{equation}
    where the Lipschitz constant \(L_\Upsilon\) is given in Lemma~\ref{Lemma_Lipschitz}.
\end{corollary}
\textcolor{black}{Corollary~\ref{Col:true_Final_Bound} extends the previous result to the true risk and reveals three distinct sources of estimation error. The first two terms correspond to finite-sample generalization effects and scale as $\mathcal{O}\left(1/\sqrt{MN}\right)$, quantifying how data availability reduces uncertainty in the learned observer. The final term, proportional to $t_s$, captures the discretization error introduced by numerical approximation of the observer dynamics. Although the bound is conservative due to worst-case constants, this decomposition highlights an explicit trade-off between data richness and numerical resolution, where increasing the number or length of trajectories improves statistical accuracy, while reducing the sampling time decreases approximation error.}

\section{Numerical case studies}\label{Sec:CaseStudies}
All the Python codes needed to reproduce the simulations and
figures are available at GitHub repository: \url{https://github.com/mwoelk236/NLOX}. 
\subsection{Bioreactor system}
We first consider the following constant volume bioreactor example \citep{gauthier1992simple}:
\begin{equation}\label{Bioreactor Model}
\begin{aligned}
    \dot{x}_1 &= \mu\left(x_1,x_2\right)x_1-ux_1,\\
    \dot{x}_2 &= -\mu\left(x_1,x_2\right)x_1 +u\left(0.1-x_2\right),\\
     y&=x_1,
\end{aligned}
\end{equation}
\noindent where $x_1$ and $x_2$ are the concentrations of the microorganism and substrate, respectively, and $u$ is a positively bounded input. The growth rate $\mu$ is given by the ``Contois'' model:
\begin{equation*}\label{Contois Model}
    \mu (x_1,x_2)=\frac{x_2}{x_1+x_2}.
\end{equation*}

\noindent  This system is simulated with a sampling time $t_s$ of $0.1$ seconds for a total of $100$ seconds. We sample $M=300$ initial points uniformly in \textcolor{black}{$\mathcal{X} = [0.05,0.1]\times[0.05,0.1]$}. The input signal is drawn from a log-normal distribution $u(t)\sim\text{log}\mathcal{N}(0.4,0.2)$, ensuring that $u(t)>0$ for all time, while providing sufficient persistent excitation. Figure \ref{fig:bioreactor_samples} shows the resulting trajectories after preprocessing, which involved splitting the data into a 70:30 training-test ratio, shifting into deviation variables, and applying min-max normalization. The observer states were initialized at zero $\left(z_0^{(i)}=0\right)$, and the observer matrices were defined as $A =\text{diag}(-3,-6)$ and $B = [1, 1]^\top$, giving an observer state dimension of $d_z=2$.

\begin{figure}[ht]
	\centering
	\includegraphics[width=10cm, height=6cm]{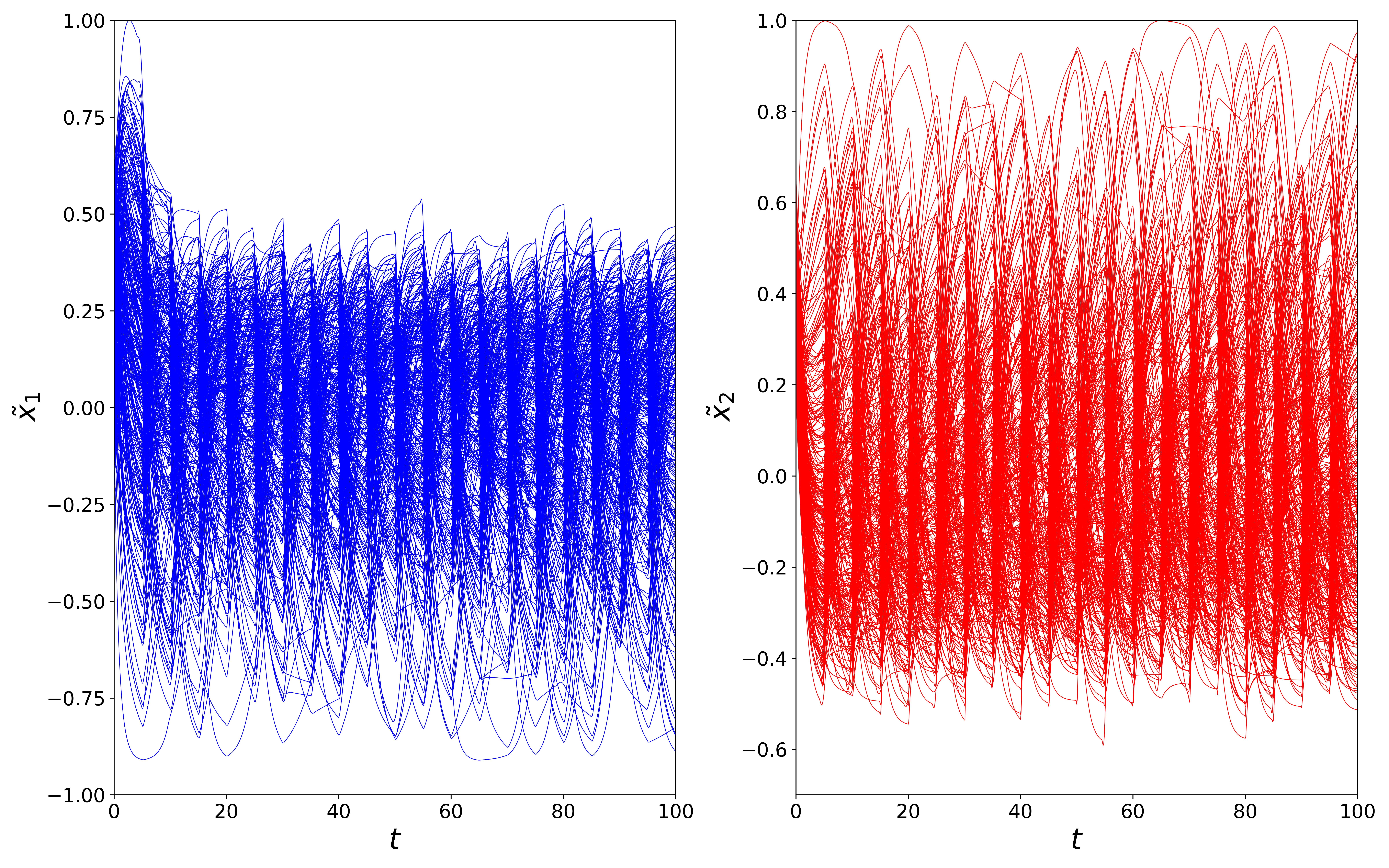}
	\caption{State trajectories for $\tilde{x}_1$ and $\tilde{x}_2$ for the bioreactor system.}\label{fig:bioreactor_samples}
\end{figure}
Following the initialization, the observed states are fed into the first feedforward MLP, $\omega(z|\theta)$, to estimate the drift dynamics. This \textcolor{black}{final network architecture used in the reported experiments} consists of 3 hidden layers ($D_\omega=4$ total layers) with 48 neurons each ($d_{\ell,\omega}=48$). The output of this network of dimensions 2 $(d_{\text{out},\omega}=2)$ is then used in \eqref{eq:KKL_discrete} to compute an updated observed state $z$ that incorporates the effect of the manipulated inputs. This updated state is subsequently passed to a second MLP, $T^\dagger(z|\vartheta)$, which is constructed with an identical architecture ($D_{T^\dagger}=4$, $d_{\ell,T^\dagger}=48$) to learn the inverse mapping to the original state $x$. The prediction error is calculated via the loss function \eqref{eq:NN_Loss}, and the parameters $(\theta, \vartheta)$ of both networks are updated via backpropagation with a learning rate $\eta=10^{-5}$. This training procedure is performed for $E=200$ epochs.

\textcolor{black}{This} final network architecture \textcolor{black}{was} determined through a grid search for the optimal network sizes, were both networks were constrained to be symmetric, $(\text{i.e.,}~D_{\omega}=D_{T^\dagger}=4,~~\text{and}~d_{\ell,\omega}=d_{\ell,T^\dagger})$. We evaluated performance by varying the number of neurons per layer from the set $\{16, 32, 48\}$; The results of this architecture search, illustrated in \textcolor{black}{Table} \ref{tab:grid_search_summary}, show that the observer's ability to predict the succeeding state improved as the number of neurons per layer increased. Furthermore, we investigated the impact of the amount of training data by training the model on $\{100, 200, 300\}$ trajectories. The results, shown in \textcolor{black}{Table} \ref{tab:grid_search_summary}, demonstrate that observer accuracy improves with more data; however, this comes at the cost of a longer training phase. \textcolor{black}{For the optimal configuration with 48 neurons per hidden layer, each network contained 4,946 trainable parameters (9,892 parameters in total across two networks). Training on 300 trajectories (approximately 300,000 samples) required 4,127.5 seconds on a single CPU core (using the high-performance computing cluster at NC State University).} The root mean squared error (RMSE) between the estimated test trajectory and the true state for these grid searches is summarized in Table~\ref{tab:grid_search_summary}.

\begin{table}[ht]
\centering
\caption{Summary of grid search results for (a) network size and (b) training data size for the bioreactor system.}
\label{tab:grid_search_summary}
\begin{minipage}{0.48\textwidth}
\centering
\subcaption{Architecture search (300 training trajectories)}
\label{tab:arch_search}
\begin{tabular}{lcc}
\toprule
\textbf{Neurons/Layer} $\left(d_{\ell,j}\right)$ & \multicolumn{2}{c}{\textbf{Test RMSE}} \\
 & $\mathbf{\tilde{x}_1}$ & $\mathbf{\tilde{x}_2}$ \\
\midrule
16 & 0.0174 & 0.0365 \\
32 & 0.0172 & 0.0242 \\
48 & 0.0159 & 0.0220 \\
\bottomrule
\end{tabular}
\end{minipage}
\hfill
\begin{minipage}{0.48\textwidth}
\centering
\subcaption{Data ablation (64 neurons/layer)}
\label{tab:data_ablation_tab}
\begin{tabular}{lcc}
\toprule
\textbf{Trajectories} $(M)$ & \multicolumn{2}{c}{\textbf{Test RMSE}} \\
 & $\mathbf{\tilde{x}_1}$ & $\mathbf{\tilde{x}_2}$ \\
\midrule
100 & 0.0325 & 0.0629 \\
200 & 0.0230 & 0.0594 \\
300 &  0.0158 & 0.0301 \\
\bottomrule
\end{tabular}
\end{minipage}
\end{table}

Using the optimal parameters, \textcolor{black}{48} neurons per layer for both networks ($d_{\ell,\omega}=d_{\ell,T^\dagger}=48$) and 300 training trajectories, a final test trial  conducted. The performance of the proposed NLOX method is compared against the analytical solution and a model-based state estimation approach in Figure \ref{fig:bioreactor_compare}, with the resulting RMSE and total root of sum of squares error (RSSE) values detailed in Table \ref{tab:bioreactor_compare}. The results demonstrate that the proposed NLOX method achieves performance comparable to all model-based benchmarks, despite operating in a completely model-free framework. A detailed comparison and further conclusions are provided in the following sections.
\begin{table}[ht]
\centering
\caption{Summary of test RMSE and RSSE results for the bioreactor system, comparing NLOX with an observer form from \cite{bernard2018luenberger}, an extended Kalman filter (EKF)\textcolor{black}{, and a sliding mode observer (SMO).}}
\label{tab:bioreactor_compare}
\begin{minipage}{0.48\textwidth}
\centering
\begin{tabular}{lccC}
\toprule
\textbf{Method}  & \multicolumn{2}{c}{\textbf{Test RMSE}}&\textbf{Test RSSE} \\
 & $\mathbf{\tilde{x}_1}$ & $\mathbf{\tilde{x}_2}$ \\
\midrule
\cite{bernard2018luenberger} & 0.0102 &0.0119& 0.0157 \\
EKF & 0.0202 & 0.0145 & 0.0249\\
SMO & 0.0199 & 0.0203 & 0.0284\\
NLOX &0.0157  & 0.0172& 0.0233\\
\bottomrule
\end{tabular}
\end{minipage}
\end{table}
\begin{figure}[h!]
	\centering
	\includegraphics[width=10cm, height=6cm]{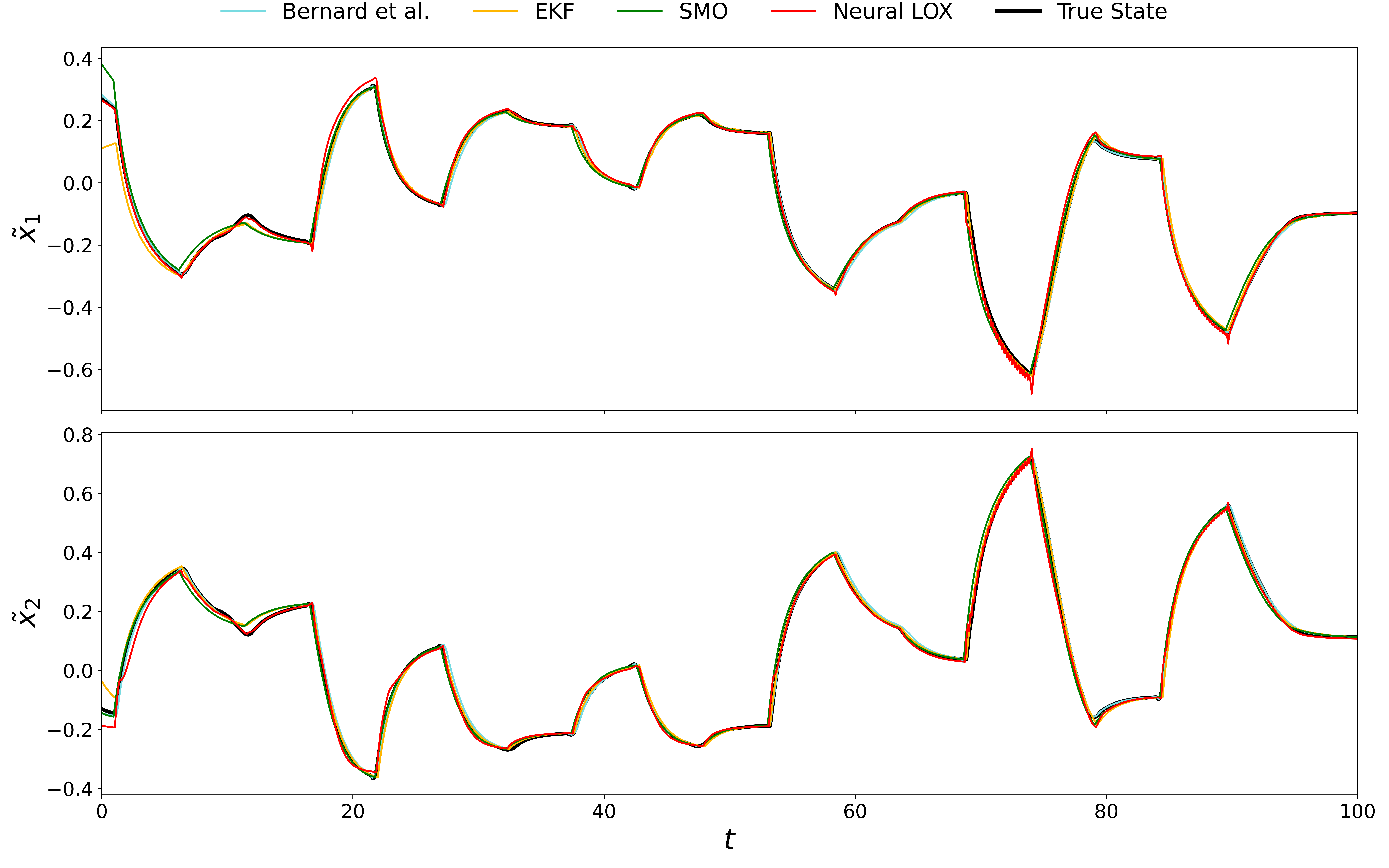}
	\caption{Comparison of state estimates for the bioreactor system: NLOX (red) versus the analytical solution (light blue), EKF (yellow), and \textcolor{black}{SMO (green).}}\label{fig:bioreactor_compare}
\end{figure}

\subsubsection{Comparison to nonlinear Luenberger observer}
In this subsection, we will consider the comparison to the analytical solution, or the said ``ground truth'' to system \eqref{Bioreactor Model} using the model-based nonlinear Luenberger observer proposed in \cite{bernard2018luenberger}. This type of observer applies to this system because it is uniformly instantaneously observable on the set:
\begin{equation*}
    \mathcal{S} = \left\{ \left(x_1,x_2\right) \in \mathbb{R}^2 \colon x_1 >0, x_2> 0\right\}.
\end{equation*}
 The drift of the system is strongly differentiable of order 2 on $\mathcal{S}$ and the input $u$ is bounded so that all the trajectories are bounded. \textcolor{black}{These properties satisfy the conditions described in \cite{bernard2018luenberger}, allowing the observer to be written in explicit form. To do so, we consider the transformation $T$ that solves the PDE \eqref{eq:KKL-PDE}. Following the approach in \cite{bernard2018luenberger}, and exploiting the structure of the bioreactor dynamics, the PDE reduces to}
 \begin{equation*}
    \frac{\partial T_\lambda \left(x_1,\xi\right)}{\partial x_1} \mu\left(x_1,\xi-x_1\right)x_1 =-\textcolor{black}{\kappa}\lambda T_\lambda \left(x_1,\xi \right)+x_1.
\end{equation*}

\noindent \textcolor{black}{Here, the quantity $\xi = x_1 + x_2$ is constant along the drift dynamics, allowing $T$ to be expressed as a function of $(x_1,\xi)$ instead of $(x_1,x_2)$, and we define $\lambda^* := \textcolor{black}{\kappa}\lambda$. This facilitates obtaining an explicit solution. Integrating with respect to $x_1$ yields a possible solution:}
\begin{equation*}
    T_{\lambda^*}(x_1,\xi) = \int_0^{x_1} \left[ \left( \frac{\xi-x_1}{\xi-s}\right) \left( \frac{s}{x_1}\right) \right]^{\lambda^* }\left(1+ \frac{1}{\xi-s}\right)~ ds.
\end{equation*}

\noindent We then define $T\left(x_1,x_2\right) = \left( T_{\lambda^*_1}\left(x_1,x_1+x_2\right),T_{\lambda^*_2}\left(x_1,x_1+x_2\right) \right)$, where $\lambda^*_1 = 3$ and $\lambda^*_2 = 6$, \textcolor{black}{corresponding to $\lambda_1 = 1$, $\lambda_2 = 2$, and $\kappa = 3$ in \cite{bernard2018luenberger}.}  The observer dynamics take the form:
\begin{equation*}\label{eq:Bernard_z_sol}
    \dot{z} = 3A z+By+\frac{dT}{dx}(\hat{x})g(\hat{x})u, \quad \hat{x}=T^{-1}(z)
\end{equation*}

\noindent which may be transformed to the x-coordinates as 
\begin{equation}\label{eq:Bernard_x_sol}
    \dot{\hat{x}} = f(\hat{x})+g(\hat{x})u+\left( \frac{dT}{dx}(\hat{x})\right)^{-1}\left( y-\hat{x}_1\right).
\end{equation}

We implement this analytical estimator \eqref{eq:Bernard_x_sol} to establish a performance baseline. The results, shown in Figure \ref{fig:bioreactor_compare} and Table \ref{tab:bioreactor_compare}, confirm that this explicit solution achieves the lowest estimation error, particularly for $\tilde{x}_1$, where it performs approximately 1.5 times better than NLOX, EKF, and SMO. In contrast, its performance in estimating $\tilde{x}_2$ is comparable to that of NLOX, EKF, and SMO.

\subsubsection{Comparison to extended Kalman filter}
In this subsection, we benchmark our observer against the model-based extended Kalman filter (EKF), the \textit{de facto} standard for nonlinear state observation. The EKF extends classical Kalman filtering to nonlinear systems through local linearization about the current state estimate \citep{khalil_nonlinear_2002}. Its state observation follows:
\begin{equation}\label{eq:KalmanObs}
\dot{\hat{x}}(t) = f(\hat{x}(t)) + g(\hat{x}(t))u(t) + L(\hat{x}(t))(y(t) - h(\hat{x}(t))),
\end{equation}
where the Kalman gain $L(t)\in\R^{n_x\times n_y}$ is defined as follows:
\begin{equation*}\label{eq:KalmanGain}
    \begin{aligned}
        L(t)=P(t)C^\top(t)R^{-1}, 
    \end{aligned}
\end{equation*}
with $P(t)$ satisfying the differential Riccati equation:
\begin{equation}
    \begin{aligned}
        \dot{P}(t)=A(t)P(t)+P(t)A^\top(t)+Q-P(t)C(t)^\top R^{-1}C(t)P(t). \label{eq:DRE}
    \end{aligned}
\end{equation}
The design requires a priori selection of constant positive definite matrices $Q \succ 0, ~R\succ0$, along with a positive definite initial condition $P(0) \succ 0$. The time-varying matrices $A(t)$ and $C(t)$ are determined by linearization at the current state estimate $\hat{x}(t)$: 
\begin{equation*}\label{jacobians}
\begin{aligned}
A(t) := \left.\frac{\partial f}{\partial x}\right|_{\hat{x}(t)} + u(t)\left.\frac{\partial g}{\partial x}\right|_{\hat{x}(t)},  \qquad C(t) := \left.\frac{\partial h}{\partial x}\right|_{\hat{x}(t)} 
\end{aligned}
\end{equation*}

\noindent Thus, \eqref{eq:KalmanObs} and \eqref{eq:DRE} must be solved  simultaneously. The extended Kalman Filter (EKF) is applied to the bioreactor model \eqref{Bioreactor Model} for state observation. For this implementation, the linearized matrices are defined as:
\begin{equation*}\label{eq:ABC}
    A(t)=\begin{bmatrix}
        \frac{\hat{x}_2\left(\hat{x}_1+\hat{x}_2 \right)-\hat{x}_1\hat{x}_2}{\left(\hat{x}_1+\hat{x}_2\right)^2}-u& \frac{\hat{x}_1\left(\hat{x}_1+\hat{x}_2 \right)-\hat{x}_1\hat{x}_2}{\left(\hat{x}_1+\hat{x}_2\right)^2}\\\frac{-\hat{x}_2\left(\hat{x}_1+\hat{x}_2 \right)+\hat{x}_1\hat{x}_2}{\left(\hat{x}_1+\hat{x}_2\right)^2}& \frac{-\hat{x}_1\left(\hat{x}_1+\hat{x}_2 \right)+\hat{x}_1\hat{x}_2}{\left(\hat{x}_1+\hat{x}_2\right)^2}-u
    \end{bmatrix},\qquad C=\begin{bmatrix}
        1&0
    \end{bmatrix}.
\end{equation*}
\textcolor{black}{The covariance matrices are chosen as $Q=\begin{bmatrix}100 & 0 \\ 0 & 1 \end{bmatrix}$, $R=0.1$, with $P_0=I_2$. These values were selected via a logarithmic grid search, yielding a preference for larger $Q$ and smaller $R$, consistent with the deterministic setting.} The system initial conditions for $(x_1,x_2)$ and $(\hat{x}_1,\hat{x}_2)$ were sampled consistently with the simulator configuration. The results in Figure \ref{fig:bioreactor_compare} and Table \ref{tab:bioreactor_compare} show that the model-based EKF has the second largest estimation error among the three approaches. Although the EKF can perform well when provided with an accurate model, its performance is primarily limited by its sensitivity to initialization. Therefore, we claim that the neural network-based data-driven observer achieves satisfactory state observation performance when a nonlinear accurate model is unavailable, outperforming the local linearization-based methodology.

\subsubsection{Comparison to sliding mode observer}
\textcolor{black}{In this subsection, we compare the NLOX method to a sliding mode observer (SMO). A sliding mode observer uses discontinuous correction terms in the state estimate dynamics to drive estimation error to zero in finite time in the presence of bounded disturbances \citep{drakunov1995sliding}. The observer estimates then converge to and evolve on the sliding manifold defined by $y-h(x)=0$, resulting in reduced-order error dynamics. A SMO has the following form:
\begin{equation}\label{SMO dynamics}
    \dot{\hat{x}}(t)=f\left(\hat{x}(t)\right)+g\left(\hat{x}(t)\right)u(t)+L'\operatorname{sign}\left(y(t)-h\left(\hat{x}(t)\right)\right)
\end{equation}
In practice, the discontinuous sign function is approximated by a continuous function (e.g., $\tanh\left(\tfrac{\cdot}{\epsilon}\right)$ for some $\epsilon\in\mathbb{R}_{>0}$) to improve numerical stability and reduce chattering. The gain matrix $L'\in\mathbb{R}^{n_x\times n_y}$ must be chosen a priori to dominate disturbances, which can be conservative and problem-dependent.}

\textcolor{black}{To alleviate the sensitivity of SMOs to this selection of a gain matrix, we consider an adaptive scaling gain matrix. The gain matrix is scaled by an adaptive parameter $\rho\in\mathbb{R}^{n_y}$ that evolves based on the error dynamics. The observer dynamics then become
\begin{equation}\label{eq:SMO_dynamics}
    \dot{\hat{x}}(t)=f\left(\hat{x}(t)\right)+g\left(\hat{x}(t)\right)u(t)+L'\left(\rho(t)\odot \tanh\left(\frac{y-h\left(\hat{x}(t)\right)}{\epsilon}\right)\right),
\end{equation}
where $\odot$ is the Hadamard product. In our implementation, the adaptive parameter $\rho$ is chosen as $\rho(t)=\left|\left(y(t)-h\left(\hat{x}(t)\right)\right)\right|$, which for the bioreactor system is $\rho(t)=\left\lvert x_1(t)-\hat{x}_1(t)\right\rvert$. The gain matrix $L'$ used for the bioreactor system is $L'=
\begin{bmatrix}
    2 & 2
\end{bmatrix}^T$. }

\textcolor{black}{The results shown in Figure \ref{fig:bioreactor_compare} and errors summarized in Table \ref{tab:bioreactor_compare} indicate that the SMO exhibits higher estimation error relative to the other approaches. While SMO can perform effectively when an accurate model and carefully tuned gains are available, its performance in this setting is less competitive. In contrast, the proposed data-driven NLOX framework achieves improved estimation accuracy without relying on an explicit model.}

\subsubsection{Sensitivity to disturbances}
\textcolor{black}{We now consider adding disturbances to \eqref{Bioreactor Model} and redefine the system as follows:
\begin{equation}\label{Bioreactor Model Noise}
\begin{aligned}
    \dot{x}_1 &= \mu\left(x_1,x_2\right)x_1-\left(u+d_1\right)x_1,\\
    \dot{x}_2 &= -\mu\left(x_1,x_2\right)x_1 +\left(u+d_1\right)\left(0.1-x_2\right),\\
     y&=x_1,
\end{aligned}
\end{equation}
where the disturbance $d_1$ is sampled from $\mathcal{N}(0.0,0.01)$. The system is simulated using the same sampling time as in the nominal case. The same training procedure described previously is applied. Through grid search, the optimal architecture in the presence of disturbances was found to be $D_{\omega}=D_{T^\dagger}=4$ layers with $d_{\ell,\omega}=d_{\ell,T^\dagger}=48$ neurons per hidden layer. After training both networks, the observer is evaluated on a test trajectory and compared against the analytical observer from \cite{bernard2018luenberger}, EKF, and SMO.}
\begin{table}[ht]
\centering
\caption{\textcolor{black}{Summary of test RMSE and RSSE results for the bioreactor system with noise, comparing NLOX with an observer form from \cite{bernard2018luenberger}, the extended Kalman filter (EKF), and a slide mode observer (SMO).}}
\label{tab:NOISE_bioreactor_compare}
\begin{minipage}{0.48\textwidth}
\centering
\begin{tabular}{lccC}
\toprule
\textbf{Method}  & \multicolumn{2}{c}{\textbf{Test RMSE}}&\textbf{Test RSSE} \\
 & $\mathbf{\tilde{x}_1}$ & $\mathbf{\tilde{x}_2}$ \\
\midrule
\cite{bernard2018luenberger} & 0.0133 &0.0156& 0.0205 \\
EKF & 0.0218 & 0.0198 & 0.0294\\
SMO & 0.00970 & 0.0128 & 0.0161\\
NLOX &0.0083  & 0.0261& 0.0274\\
\bottomrule
\end{tabular}
\end{minipage}
\end{table}
\begin{figure}[h!]
	\centering
	\includegraphics[width=10cm, height=6cm]{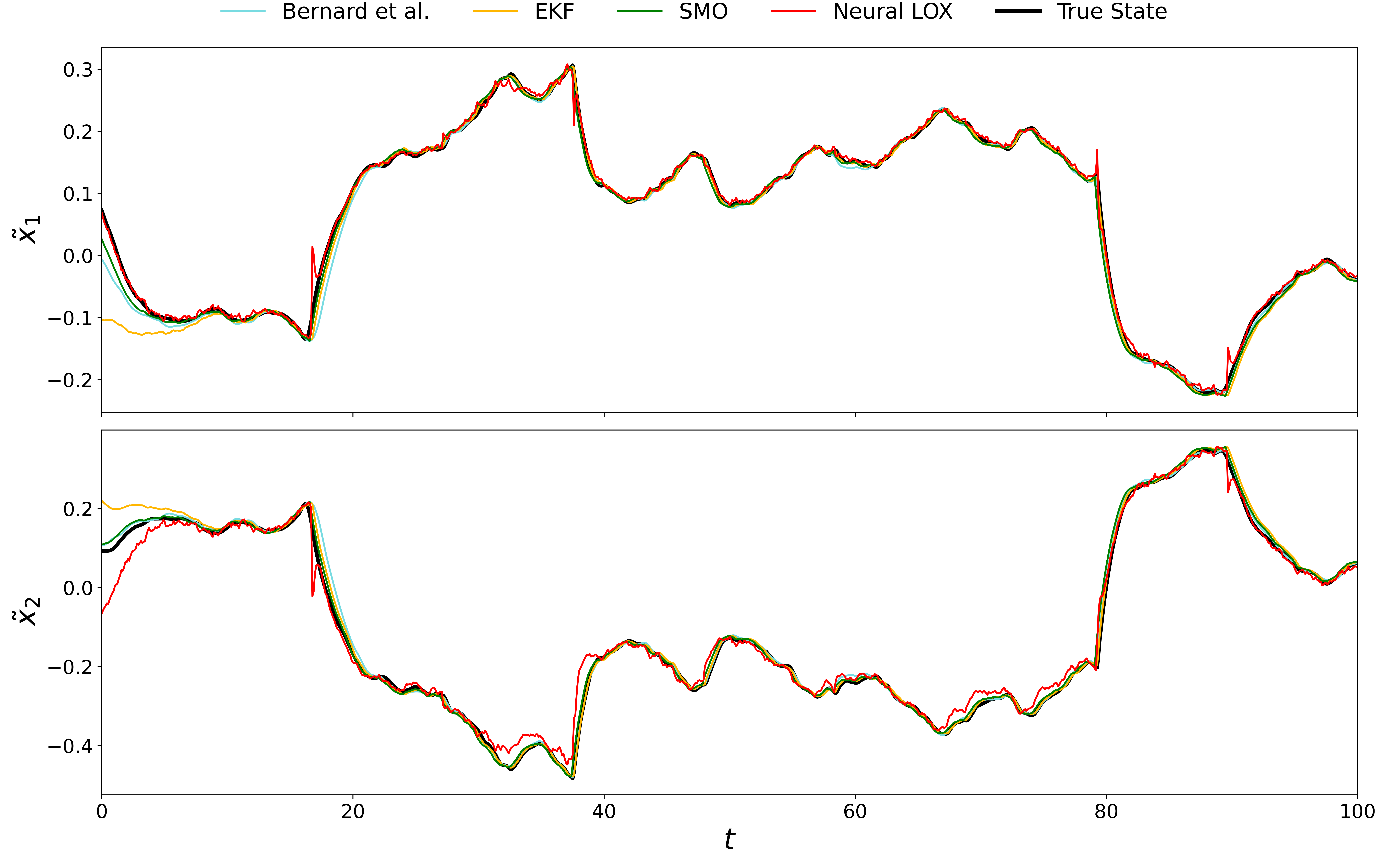}
	\caption{\textcolor{black}{Comparison of state estimates for the bioreactor system with noise: NLOX (red) versus the analytical solution (light blue), EKF (yellow), and SMO (green).}}\label{fig:bioreactor_compare_NOISE}
\end{figure}

\textcolor{black}{The results are summarized in Table~\ref{tab:NOISE_bioreactor_compare}, and a test state trajectory is shown in Figure~\ref{fig:bioreactor_compare_NOISE}. The results indicate that the proposed NLOX observer maintains competitive performance under disturbance conditions. In particular, NLOX achieves lower estimation error for $\tilde{x}_1$ compared to the benchmark methods, while exhibiting comparable performance for $\tilde{x}_2$. These results demonstrate that the proposed data-driven observer retains robust estimation behavior in the presence of stochastic input perturbations.}

\subsection{Williams-Otto reactor }
Next, we consider the Williams-Otto reactor, a nonisothermal continuously stirred-tank reactor (CSTR) with three reactions. \citep{de2020simultaneous}.
\begin{equation*}\label{Williams-Otto-Rxns}
\begin{aligned}
    \ce{\textit{A} + \textit{B} &->[k_1] \textit{C}} \\
    \ce{\textit{B} + \textit{C} &->[k_2] \textit{P} + \textit{E}} \\
    \ce{\textit{C} + \textit{P} &->[k_3] \textit{G}}
\end{aligned}
\end{equation*}
The reactor is fed two pure component reactant flow rates, $F_A$ and $F_B$, and the output flow rate is composed of all 6 components involved in the reactions above. Component $P$ is the desired product of the system, $E$ is a byproduct, $C$ is an intermediate, and $G$ is the residual material. The transient mass balance on these six components yields a state-space representation of this reactor system.
\begin{subequations}\label{Williams-Otto-Main}
    \begin{align}
        \dot{x}_A&=\frac{F_A-\left(F_A+F_B\right)x_A}{W}-k_1x_Ax_B,\\
        \dot{x}_B&=\frac{F_B-\left(F_A+F_B\right)x_B}{W}-k_1x_Ax_B-k_2x_Bx_C,\\
        \dot{x}_C&=\frac{-\left(F_A+F_B\right)x_C}{W}+2k_1x_Ax_B-2k_2x_Bx_C-k_3x_Cx_P,\\
        \dot{x}_E&=\frac{-\left(F_A+F_B\right)x_E}{W}+2k_2x_Bx_C,\\
        \dot{x}_G&=\frac{-\left(F_A+F_B\right)x_G}{W}+\frac{3}{2}k_3x_Cx_P,\\
        \dot{x}_P&=\frac{-\left(F_A+F_B\right)x_P}{W}+k_2x_Bx_C-\frac{1}{2}k_3x_Cx_P.
    \end{align}
\end{subequations}
The vector of these six compositions is the state vector of the system, and the reaction constants $k_1,~k_2,$ and $k_3$ $(\text{s}^{-1})$ are functions of the reactor temperature $(T_R)~({\text\textdegree} \text{C})$ and are defined by the Arrhenius equation as follows:
\begin{equation*}
    \begin{aligned}
        k_1 &= 1.6599\times10^6\cdot \exp \left[\frac{-6666.67}{T_R+273.15} \right],\\
        k_2 &= 7.2117\times10^8\cdot \exp \left[\frac{-8333.33}{T_R+273.15} \right],\\
        k_3 &= 2.6745\times10^{12}\cdot \exp \left[\frac{-11111}{T_R+273.15} \right],\\
    \end{aligned}
\end{equation*}
The measured outputs in this system are the compositions of the desired product $P$ and the less desired byproduct $E$ $\left(y=\left[x_E,x_P\right]\right)$. The manipulated inputs for this system are the reactor temperature and the input flow rate of component $B$ $\left(u=\left[F_B,T_R\right]\right)$, as according to the literature \citep{amrit2013optimizing}. The characteristic time of this system is roughly the residence time for the reactor, $\tau:=W/\left( F_A+F_B\right)$. Thus, \eqref{Williams-Otto-Main} is scaled by the residence time using a dimensionless time $\tau:=t/\tau^*$. The compositions of each component were also shifted to deviation variables, where the origin is the equilibrium point, denoted $\tilde{x}_i$. This system is simulated for $M=200$ trajectories with a sampling time of $0.1$ residence times. The inputs are randomly changed every  $5\tau$ according to the normal distributions shown in Table \ref{ottotable}. The reactor parameters and initial conditions of $x_i$ are summarized in Table \ref{ottotable}, and a sampled system trajectory is shown in Fig. \ref{fig:CSTR}. 

\begin{table}[width=.9\linewidth,cols=4,pos=h]
\caption{Parameters and initial conditions of the Williams-Otto reactor}\label{ottotable}
\begin{tabular*}{\tblwidth}{@{} LLLL@{} }
\midrule
$F_A=3.5$ kg/s & $x_A(0) \sim \mathcal{U} \left[0.2,0.6\right] $   \\
$F_B\sim \mathcal{N}(6.0,1.0)$ kg/s & $x_B(0)=1-x_A(0)$ \\
$T_R\sim \mathcal{N}(100.0,15.0)$ \textdegree C & $x_C(0)=0$ \\
$W=2500$ kg  & $x_E(0)=0$ \\
$\tau=263.2~\text{s}$ & $x_G(0) = 0 $\\
& $x_P(0)=0$\\
\bottomrule
\end{tabular*}
\end{table}

\begin{figure}[ht]
\centering
\includegraphics[width=10cm, height=6cm]{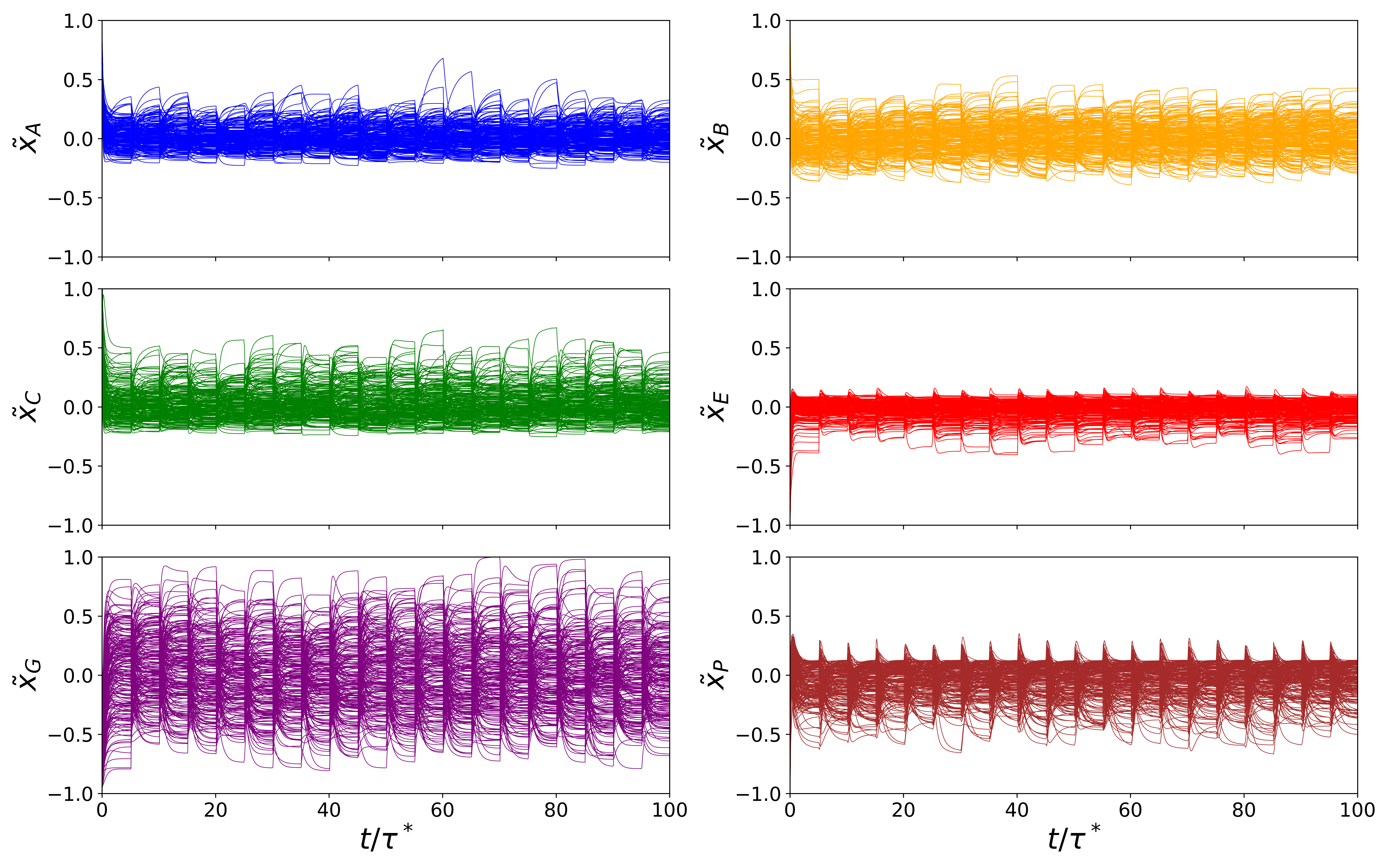}
\caption{Training state trajectories for Williams-Otto reactor.}
\label{fig:CSTR}
\end{figure}
\noindent The observer system matrices are chosen as $A =\text{diag}(-1,-2,\dots,-14)$ and $B = \mathbf{1}_{14 \times 2}$ (a $14 \times 2$ matrix of ones), with the observer state dimension $d_z=14$, and the observer states are initialized at zero ($z_0^{(i)}=0$).  The higher observer dimension is selected to capture the nonlinearities and larger state dimension of the Williams-Otto reactor. 

Following the initialization, the observed states are fed into the first feedforward MLP, $\omega(z|\theta)$, which has 3 hidden layers ($D_\omega=4$ total layers) with 64 neurons each ($d_{\ell,\omega}=64$), and an output dimensions of 28 $(d_{\text{out},\omega}=28)$. This estimate is then used in \eqref{eq:KKL_discrete} to update the observed state $z$ that incorporates the effect of the manipulated inputs. 

This updated state is subsequently passed to the second stage for learning the inverse mapping. Instead of training a single network to learn the full mapping $T^\dagger:\R^{14}\to\R^6$, we train six separate networks, each dedicated to estimating one of the original six states. While the 14-dimensional observed states $z$ are sufficient to contain the full system information, using a single network with one objective function failed to converge during preliminary tests. We attribute this to two main factors. First, training a single network on a composite loss function creates competing gradients; the backpropagation update direction that improves the prediction for one state can negatively impact others. Second, the nonconvex loss landscape makes it difficult for a single network to escape poor local minima that fail to accurately capture the dynamics of all six states. Thus, decomposing the problem into six simpler learning tasks helps to resolve this convergence issue. Each of the six networks comprising $T^\dagger$ is constructed with an identical architecture to the $\omega$ network ($D_{T^\dagger}=4$, $d_{\ell,T^\dagger}=64$). The prediction error for each state is calculated via the loss function \eqref{eq:NN_Loss}, and the parameters $(\theta, \vartheta)$ of all networks are updated simultaneously via backpropagation with a learning rate of $\eta=10^{-6}$. This training procedure is performed for $E=300$ epochs.

The network architecture is selected via a grid search to optimize performance. All networks were symmetric, with $D_{\omega}=D_{T^\dagger}=4$ layers and $d_{\ell,\omega}=d_{\ell,T^\dagger}$ neurons per layer. We evaluated architectures with ${16, 32, 48, 64}$ neurons per layer; as shown in \textcolor{black}{Table} \ref{tab:grid_search_summary_otto}, prediction accuracy improved with increasing network size. We also assessed the impact of training data volume using ${50, 100, 200}$ trajectories. \textcolor{black}{Table} \ref{tab:grid_search_summary_otto} shows that while more data improves observer accuracy, it also increases training time. The RMSEs for these grid searches are summarized in Table~\ref{tab:grid_search_summary_otto}.

\begin{table}[ht]
\centering
\caption{Summary of grid search results for the Williams-Otto reactor.}
\label{tab:grid_search_summary_otto}
\begin{subtable}{\textwidth}
\centering
\subcaption{Architecture search (300 training trajectories)}
\label{tab:arch_search_otto}
\begin{tabular}{lcccccc}
\toprule
\textbf{Neurons/Layer} $\left(d_{\ell,j}\right)$ & \multicolumn{6}{c}{\textbf{Test RMSE}} \\
 & $\mathbf{\tilde{x}_A}$ & $\mathbf{\tilde{x}_B}$ & $\mathbf{\tilde{x}_C}$ & $\mathbf{\tilde{x}_E}$ & $\mathbf{\tilde{x}_G}$ & $\mathbf{\tilde{x}_P}$ \\
\midrule
16 & 0.0335 & 0.0362 & 0.0295 & 0.0203 & 0.0561 & 0.0401 \\
32 & 0.0473 & 0.0476 & 0.0366 & 0.0309 & 0.0498 & 0.0526 \\
48 & 0.0267 & 0.0205 & 0.0352 & 0.0186 & 0.0327 & 0.0252 \\
64 & 0.0268 & 0.0205 & 0.0318 & 0.0163 & 0.0341 & 0.0230 \\
\bottomrule
\end{tabular}
\end{subtable}
\vspace{0.7cm} 
\begin{subtable}{\textwidth}
\centering
\subcaption{Data ablation (64 neurons/layer)}
\label{tab:data_ablation_tab_otto}
\begin{tabular}{lcccccc}
\toprule
\textbf{Trajectories} $(M)$ & \multicolumn{6}{c}{\textbf{Test RMSE}} \\
 & $\mathbf{\tilde{x}_A}$ & $\mathbf{\tilde{x}_B}$ & $\mathbf{\tilde{x}_C}$ & $\mathbf{\tilde{x}_E}$ & $\mathbf{\tilde{x}_G}$ & $\mathbf{\tilde{x}_P}$ \\
\midrule
50 & 0.0330 & 0.0369 & 0.0368 & 0.0354 & 0.0544 & 0.0476 \\
100 & 0.0356 & 0.0357 & 0.0491 & 0.0278 & 0.0388 & 0.0421 \\
200 & 0.0268 & 0.0205 & 0.0318 & 0.0163 & 0.0341 & 0.0230 \\
\bottomrule
\end{tabular}
\end{subtable}
\end{table}

Using the optimal parameters, 64 neurons per layer for both networks ($d_{\ell,\omega}=d_{\ell,T^\dagger}=64$) and 200 training trajectories, a final test trial is conducted. \textcolor{black}{Under this configuration, the $\omega$ network contained 11,100 trainable parameters, while the six reconstruction networks $T^\dagger$ jointly contributed 56,070 parameters, resulting in 67,170 total trainable parameters. Training on 200 trajectories required approximately 15,487 seconds on a single CPU core (using the high-performance computing cluster at NC State University).} The performance of the proposed NLOX method is compared against both model-based state estimation approaches in Figure \ref{fig:otto_compare}, with the resulting RMSE and RSSE values detailed in Table \ref{tab:otto_compare}. \textcolor{black}{The results demonstrate that the proposed NLOX method achieves a \textcolor{black}{3.19}-fold and 1.37-fold improvement in accuracy compared to the EKF and SMO benchmarks, respectively, despite operating in a completely model-free framework.}
\begin{table}[h]
\centering
\caption{Summary of test RMSE and RSSE results for the Williams-Otto reactor, comparing NLOX with the extended Kalman filter (EKF)\textcolor{black}{, and a sliding mode observer (SMO).}}
\label{tab:otto_compare}
\begin{tabular}{lccccccc}
\toprule
\textbf{Method}  & \multicolumn{6}{c}{\textbf{Test RMSE}}& \textbf{Test RSSE} \\
& $\mathbf{\tilde{x}_A}$ & $\mathbf{\tilde{x}_B}$& $\mathbf{\tilde{x}_C}$ & $\mathbf{\tilde{x}_E}$& $\mathbf{\tilde{x}_G}$ & $\mathbf{\tilde{x}_P}$ \\
\midrule
EKF & 0.0547 & 0.0457 & 0.0337 & 0.0404 & 0.0776 & 0.0454 &0.126\\
SMO & 0.0189 & 0.0165 & 0.0267 & 0.0152 & 0.0307 & 0.0204 &0.0541\\
NLOX & 0.0142 & 0.0162 & 0.0160 & 0.00608 & 0.0228 & 0.0166 &0.0394\\
\bottomrule
\end{tabular}
\end{table}

\begin{figure}[ht]
	\centering
	\includegraphics[width=10cm, height=6cm]{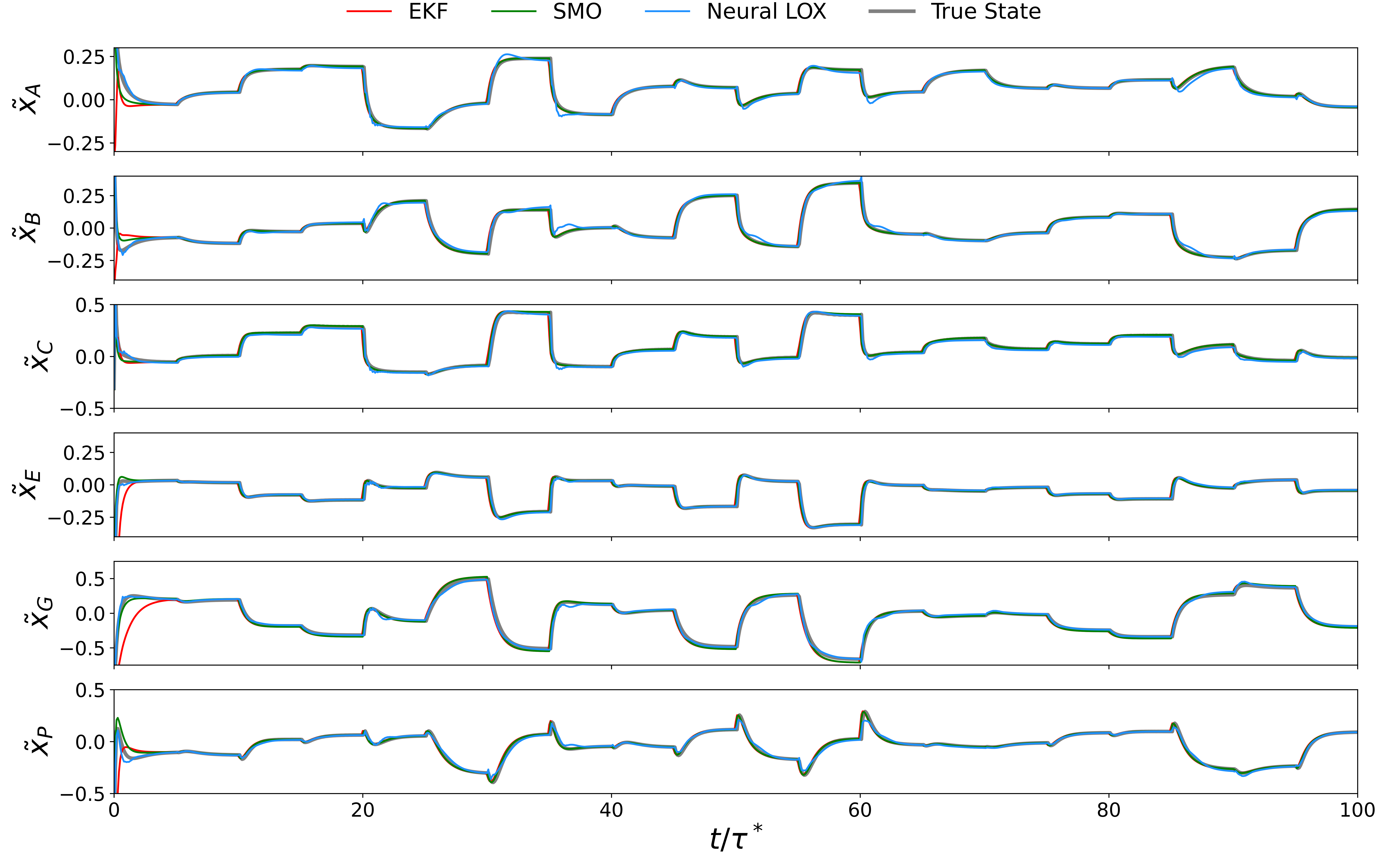}
	\caption{Comparison of state estimates for the Williams-Otto reactor: NLOX (red) versus  EKF (light blue) and \textcolor{black}{SMO (green)}.}\label{fig:otto_compare}
\end{figure}

\newpage
\subsubsection{Comparison to EKF}
Obtaining an analytical solution for the KKL observer of this higher-dimensional system is intractable. Therefore, we compare our approach against EKF as a benchmark. For the EKF implementation, the system matrix is linearized as $A(t) = A_1 + \Delta A$, where:
\begin{equation*}\label{eq:AC_system2}
\begin{aligned}
    A_1 &= -\frac{ F_A+F_B}{W}I_6,\\
    \Delta A&=\begin{bmatrix}
        -k_1\hat{x}_B&-k_1\hat{x}_A& 0&0&0&0\\
        -k_1\hat{x}_B&-k_1\hat{x}_A-k_2\hat{x}_C&-k_2\hat{x}_B&0&0&0\\
        2k_1\hat{x}_B&2k_1\hat{x}_A-2k_2\hat{x}_C&-k_2\hat{x}_B-k_3\hat{x}_P&0&0&0\\
        0&2k_2\hat{x}_C&k_2\hat{x}_B&1&0&0\\
        0&0&\frac{3}{2}k_3\hat{x}_P&0&1&\frac{3}{2}k_3\hat{x}_C\\
        0&k_2\hat{x}_C&k_2\hat{x}_B-\frac{1}{2}k_3\hat{x}_P&0&0&-\frac{1}{2}k_3\hat{x}_C
    \end{bmatrix},\\ C(t)&=\begin{bmatrix}
        0&0&0&1&0&1
\end{bmatrix},   
\end{aligned}
\end{equation*}

\noindent \textcolor{black}{and $I_6$ denotes the $6$-dimensional identity matrix. The covariance matrices are chosen as $Q=0.01 I_6$, $R=1\times10^{-4} I_2$, with $P_0 = I_6$. These values were selected via a logarithmic grid search.} \textcolor{black}{The system initial conditions for $(x_A,x_B,x_C,x_E,x_G,x_P)$ and $(\hat{x}_A,\hat{x}_B,\hat{x}_C,\hat{x}_E,\hat{x}_G,\hat{x}_P)$ were independently sampled in accordance with the simulator configuration.} As shown in Figure \ref{fig:otto_compare} and Table \ref{tab:otto_compare}, the proposed NLOX method demonstrates a lower estimation error than the EKF. NLOX's key advantage lies in its ability to maintain accurate estimation performance even during the initialization phase, where the system has not yet reached a local operating regime and the EKF unperformed. This leads to faster convergence and, consequently, a lower overall prediction error.
\subsubsection{Comparison to sliding mode observer}
\textcolor{black}{As another model-based benchmark, we compare our NLOX to a SMO. For the Williams-Otto reactor case study, we implement the SMO with the gain matrix chosen as
\begin{equation*}
    L'=
    \begin{bmatrix}
        0.2 & 0.5& 0.8& 0.2& 0.0& 0.0\\
        0.0 &0.5 & 0.8& 0.0& 0.4& 0.2
    \end{bmatrix}^\top.
\end{equation*}
This choice was heuristic but guided by the input–output structure of the model. In particular, the measured states $x_E$ and $x_P$ are corrected using only their corresponding output residuals. The adaptive parameter is chosen component-wise based on the output mismatch,
\begin{equation*}
    \rho(t)=
    \begin{bmatrix}
        \left\lvert x_E(t)-\hat{x}_E(t)\right\rvert \\
        \left\lvert x_P(t)-\hat{x}_P(t)\right\rvert
    \end{bmatrix}.
\end{equation*}so that the discontinuous term scales with the magnitude of each measurement residual.}

\textcolor{black}{The results shown in Figure \ref{fig:otto_compare} and the errors summarized in Table \ref{tab:otto_compare} demonstrate that the proposed NLOX framework achieves a 1.96-fold reduction in estimation error relative to SMO. In turn, SMO exhibits a \textcolor{black}{2.32}-fold improvement over the EKF. While SMO can yield competitive performance when an accurate process model is available and the observer gains are carefully tuned, its effectiveness remains sensitive to gain selection. In contrast, the data-driven NLOX approach provides improved estimation accuracy without requiring explicit model knowledge or manual gain tuning. }

\subsubsection{Sensitivity to disturbances}
\textcolor{black}{
We now consider the effect of disturbances by modifying \eqref{Williams-Otto-Main}, redefining the input flow rate of component B as $6.0 + d_2$, where $d_2 \sim \mathcal{N}(0.0,1.2)$. The system is simulated using the same sampling time as in the nominal case. The same training procedure described previously is applied. Through grid search, the optimal architecture under disturbance conditions was found to be $D_{\omega}=D_{T^\dagger}=4$ layers with $d_{\ell,\omega}=d_{\ell,T^\dagger}=32$ neurons per hidden layer. Compared to the noise-free setting, a network with fewer neurons provided improved generalization performance, consistent with reduced overfitting in the presence of stochastic input perturbations. After training both networks, the observer is evaluated on a test trajectory and compared against EKF and SMO.}

\begin{table}[h]
\centering
\caption{\textcolor{black}{Summary of test RMSE and RSSE results for the Williams-Otto reactor with noise, comparing NLOX with the extended Kalman filter (EKF) and a sliding mode observer (SMO).}}
\label{tab:otto_compare_noise}
\begin{tabular}{lccccccc}
\toprule
\textbf{Method}  & \multicolumn{6}{c}{\textbf{Test RMSE}}& \textbf{Test RSSE} \\
& $\mathbf{\tilde{x}_A}$ & $\mathbf{\tilde{x}_B}$& $\mathbf{\tilde{x}_C}$ & $\mathbf{\tilde{x}_E}$& $\mathbf{\tilde{x}_G}$ & $\mathbf{\tilde{x}_P}$ \\
\midrule
EKF & 0.0499 & 0.0638 &  0.0233 & 0.0386 & 0.0480 &0.0517 & 0.116\\
SMO &0.00944& 0.0223&0.0124&0.0130&0.0203&0.0146&0.0392\\
NLOX & 0.0216 & 0.0343 & 0.0146 & 0.00901 & 0.0380 & 0.0243 &0.0630\\
\bottomrule
\end{tabular}
\end{table}
\begin{figure}[ht]
	\centering
	\includegraphics[width=10cm, height=6cm]{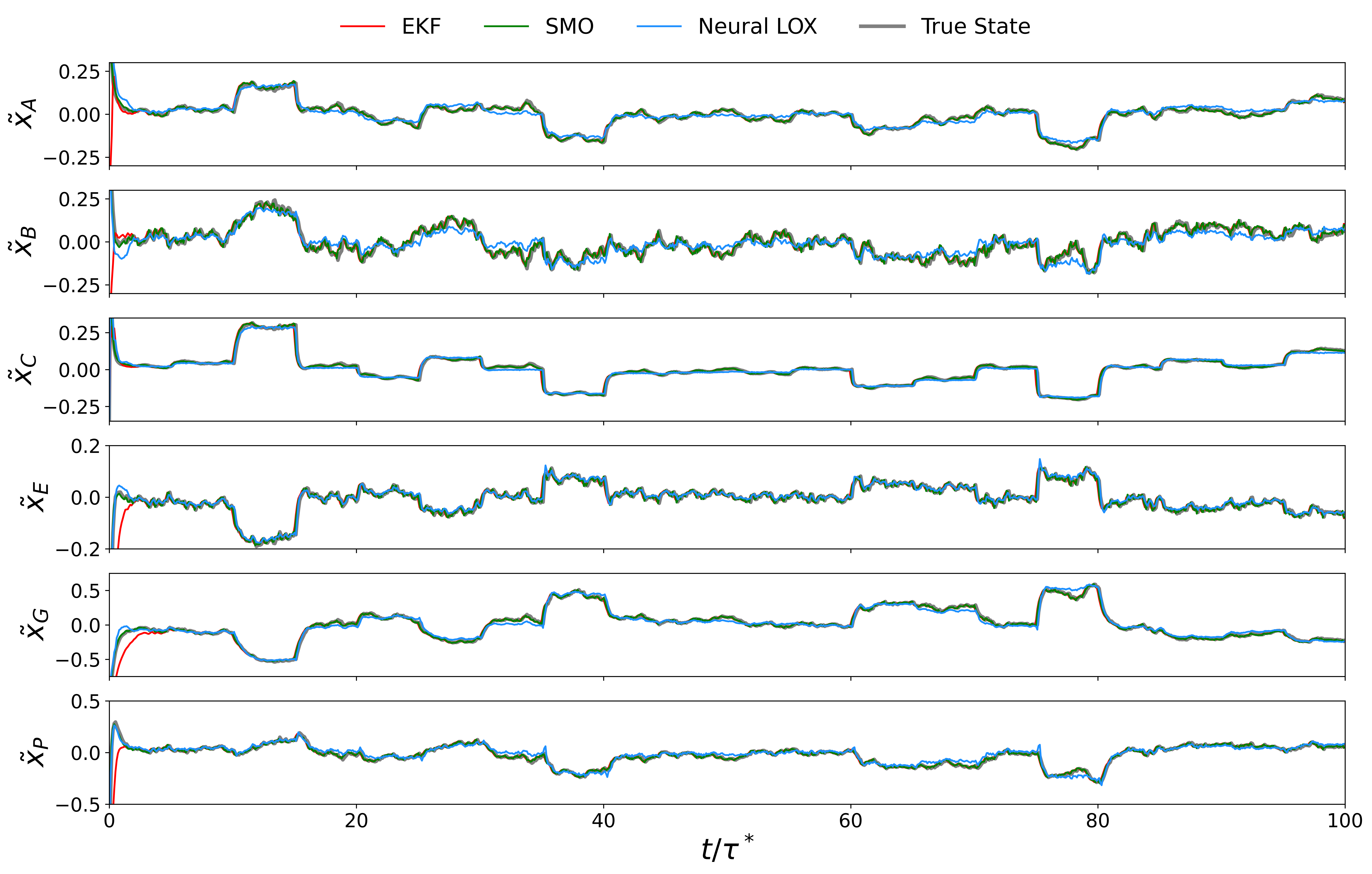}
	\caption{\textcolor{black}{Comparison of state estimates for the Williams-Otto reactor with noise: NLOX (red) versus EKF (light blue) and SMO (green)}.}\label{fig:otto_compare_noise}
\end{figure}

\par \textcolor{black}{ The results are summarized in Table~\ref{tab:otto_compare_noise}, and a test trajectory is shown in Figure~\ref{fig:otto_compare_noise}. The results indicate that the proposed NLOX observer maintains strong estimation performance under disturbance conditions. In particular, NLOX achieves consistently lower RMSE across all state variables compared to the EKF benchmark, and performs similarly to the SMO benchmark. These findings demonstrate that the proposed data-driven observer retains stable and accurate estimation behavior in the presence of stochastic input perturbations, even in this higher-dimensional setting.}

\section{Conclusion}\label{Sec:Conclusion}
This paper proposes a data-driven state observer for input-affine, nonautonomous systems. The observer follows the Kazantzis-Kravaris/Luenberger (KKL) structure, using a static injection into the observer states with an input-affine correction term. We learn both the static immersion and the input-affine term using two feedforward multilayer perceptron (MLP) neural networks. These networks are trained within an iterative scheme that estimates (i) the input-affine term and (ii) the inverse mapping of the injective mapping from system states to observer states. This approach provides a cohesive framework for estimating states directly from input and output data, without requiring online optimization or a first-principles model. The effectiveness of the proposed approach is demonstrated through two case studies, with particular attention given to its performance on higher-dimensional systems such as the Williams-Otto reactor. This addresses a significant gap in the literature, as existing data-driven observers are largely limited to, or have only been tested on, low-dimensional systems ($n_x \leq 3$). Our work demonstrates the method's efficacy in a higher-dimensional setting.

While the proposed algorithm provides a promising data-driven state observer for nonautonomous systems, the current approach relies on supervised learning, requiring access to true state data. To enhance practical applicability, future work will explore unsupervised learning methods. The core idea is to formulate a joint optimization problem where one network learns a latent state representation consistent with the input-output map, while the NLOX observer simultaneously learns to track this ``latent state'', eliminating the need for ground-truth state data.

\end{doublespace}

\printcredits 

\vspace{3mm}
\noindent {\large\textbf{Acknowledgment}}

We acknowledge the computing resources provided by North Carolina State University High Performance Computing Services Core Facility (RRID:SCR$\_$022168).
\bibliographystyle{cas-model2-names}

\bibliography{cas-refs}
\end{document}